\documentclass[10pt,twocolumn,twoside] {IEEEtran}

\usepackage{amsmath,graphicx}
\usepackage{amsthm}
\usepackage{amscd}
\usepackage{amssymb}
\usepackage{latexsym}
\usepackage{cite}
\usepackage{multirow}
\usepackage{tikz}
\usepackage{epstopdf}
\usepackage{url}
\ifCLASSOPTIONcompsoc
\usepackage[caption=false,font=normalsize,labelfont=sf,textfont=sf]{subfig}
\else
\usepackage[caption=false,font=footnotesize]{subfig}
\fi

\newtheorem{theorem}{Theorem}[section]
\newtheorem{lemma}[theorem]{Lemma}
\newtheorem{prop}[theorem]{Proposition}

\newtheorem{dfn}[theorem]{Definition}
\newtheorem{remark}[theorem]{Remark}

\newcommand{\overbar}[1]{\mkern 1.5mu\overline{\mkern-1.5mu#1\mkern-1.5mu}\mkern 1.5mu}
\def\a{\mathbf{a}}
\def\b{\mathbf{b}}
\def\d{\mathbf{d}}

\def\y{\mathbf{y}}
\def\A{\mathbf{A}}

\def\F{\mathbf{F}}

\def\f{\mathbf{f}}

\def\G{\mathbf{G}}

\def\g{\mathbf{g}}

\def\D{\mathbf{D}}
\def\mD{\mathcal{D}}
\def\one{\mathbf{1}}
\def\x{\mathbf{x}}
\def\xs{\mathbf{x}^{\star}}
\def\xr{\hat{\mathbf{x}}}

\def\z{\mathbf{z}}
\def\zs{\mathbf{z}^{\star}}
\def\zr{\hat{\mathbf{z}}}

\def\mA{\mathcal{A}}
\def\mset{\mathcal{A}}

\def\Prob{\mathbb{P}}

\def\mE{\mathbb{E}}

\def\C{\mathbb{C}}
\def\mC{\mathbf{C}}

\def\R{\mathbb{C}}
\def\mR{\mathbf{R}}

\def\I{\mathbf{I}}

\def\bigo{\mathbf{O}}
\def\S{\mathbf{S}}

\def\U{\mathbf{U}}

\def\v{\mathbf{v}}

\def\OOmega{\Omega}

\def\H{\mathbf{H}}

\def\eps{\boldsymbol{\epsilon}}
\def\bxi{\boldsymbol{\xi}}
\def\PPhi{\boldsymbol{\Phi}}
\def\PPsi{\boldsymbol{\Psi}}

\def\ric{\delta_s}
\def\sset{\mathcal{D}_{s,n}}

\def\mT{\mathcal{T}}

\def\Bset{\mathcal{B}}
\def\B{\mathbf{B}}
\def\tB{\widetilde{\mathbf{B}}}
\def\tn{\tilde{n}}

\def\xnorm{\|\x\|_2}

\def\ttwo{2\rightarrow2}

\def\diag{\text{diag}}
\def\supp{\text{supp}}

\def\TTheta{\boldsymbol{\Theta}}

\def\LLambda{\boldsymbol{\Lambda}}
\def\llambda{\boldsymbol{\lambda}}

\def\vset{\mathcal{A}_{\v}}
\def\bigo{\mathcal{O}}
\def\w{\mathbf{w}}

\def\skripc{\delta_{s,k}}

\def\r{\mathbf{r}}

\begin{document}

\title{Uniform Recovery Bounds for Structured Random Matrices in Corrupted Compressed Sensing}

\author{Peng~Zhang, Lu~Gan, Cong~Ling and~Sumei~Sun
\thanks{P. Zhang was with the Department
of Electrical and Electronic Engineering, Imperial College London, London,
SW7 2AZ, UK (e-mail: p.zhang12@imperial.ac.uk). He is now with the Institute for Infocomm Research, A$^{*}$STAR, Singapore, 138632, Singapore (e-mail: zhangp@i2r.a-star.edu.sg).}
\thanks{C. Ling is with the Department
of Electrical and Electronic Engineering, Imperial College London, London,
SW7 2AZ, UK (e-mail: cling@ieee.org) .}
\thanks{L. Gan is with the College of Engineering, Design and Physical Science, Brunel University, London, UB8 3PH, UK (e-mail: lu.gan@brunel.ac.uk).}
\thanks{S. Sun is with the Institute for Infocomm Research, A$^{*}$STAR, Singapore, 138632, Singapore (e-mail: sunsm@i2r.a-star.edu.sg).}}




\maketitle

\begin{abstract}
We study the problem of recovering an $s$-sparse signal $\xs\in\C^n$ from corrupted measurements $\y=\A\xs+\zs+\w$, where $\zs\in\C^m$ is a $k$-sparse corruption vector whose nonzero entries may be arbitrarily large and $\w\in\C^m$ is a dense noise with bounded energy. The aim is to exactly and stably recover the sparse signal with tractable optimization programs. In this paper, we prove the \emph{uniform recovery guarantee} of this problem for two classes of structured sensing matrices. The first class can be expressed as the product of a unit-norm tight frame (UTF), a random diagonal matrix and a bounded column-wise orthonormal matrix (e.g., partial random circulant matrix). When the UTF is bounded (i.e. $\mu(\U)\sim1/\sqrt{m}$), we prove that with high probability, one can recover an $s$-sparse signal exactly and stably by $l_1$ minimization programs even if the measurements are corrupted by a sparse vector, provided $m=\bigo(s\log^2s\log^2n)$ and the sparsity level $k$ of the corruption is a constant fraction of the total number of measurements. The second class considers randomly sub-sampled orthonormal matrix (e.g., random Fourier matrix). We prove the uniform recovery guarantee provided that the corruption is sparse on certain sparsifying domain. Numerous simulation results are also presented to verify and complement the theoretical results.
\end{abstract}


\section{Introduction}
The theory of compressed sensing has been widely studied and applied in various promising applications over the recent years \cite{EJC05robustuncertainty,EJC06nearoptimal,Donoho06compressedsensing,eldar2012compressed,foucart2013mathematical}. It provides an efficient way to recover a sparse signal from a relatively small number of measurements. Specifically, an $s$-sparse signal $\xs$ is measured through
\begin{align}
\y=\A\xs+\w,
\end{align}
where $\A\in\R^{m\times n}$ is referred to as the sensing matrix, $\y\in\R^{m}$ is the measurement vector and $\w\in\R^{m}$ represents the noise vector with the noise level $\|\w\|_2\leq\varepsilon$. It has been shown that if $\A$ satisfies the restricted isometry property (RIP) and $\varepsilon$ is small, the recovered signal $\xr$ obtained by $l_1$ norm minimization is close to the true $\xs$, i.e. $\|\xr-\xs\|\leq C\varepsilon$ with $C$ being a small numerical constant. Many types of sensing matrices have been proven to satisfy the RIP condition. For example, random Gaussian/Bernoulli matrices satisfy the RIP with high probability if $m\geq\bigo(s\log (n/s))$ \cite{Donoho06compressedsensing,EJC05robustuncertainty}, whereas structured sensing matrices consisting of either randomly subsampled orthonormal matrix \cite{rauhut2010compressive} or modulated unit-norm tight frames \cite{zhang2014modulated} have the RIP with high probability when $m$ is about $\bigo(s\log^4 n)$\footnote{Recent works \cite{bourgain2014improved}\cite{haviv2017restricted} for subsampled Fourier matrices show that the factor $\log^4 n$ can be reduced to $\log^3 n$.}.

This standard compressed sensing problem has been generalized to cope with the recovery of sparse signals when some unknown entries of the measurement vector $\y$ are severely corrupted. Mathematically, we have
\begin{align}
\y=\A\xs+\zs+\w,
\end{align}
where $\zs\in\R^{m}$ is an unknown sparse vector. To reconstruct $\xs$ from the measurement vector $\y$, the following penalized $l_1$ norm minimization has been proposed:
\begin{align}\label{eqn: penalized rec}
\min_{\x,\z}\|\x\|_1+\lambda\|\z\|_1\quad\quad \mathrm{s.t.} \quad\quad \|\y-(\A\x+\z)\|_2\leq\varepsilon.
\end{align}
In \cite{li2013compressed}, it was shown that random Gaussian matrices can provide uniform recovery guarantees to this problem \eqref{eqn: penalized rec}. In other words, a single random draw of a Gaussian matrix $\A$ is able to stably recover all $s$-sparse signals $\xs$ and all $k$-sparse corruptions $\zs$ simultaneously with high probability. On the other hand, for structured sensing matrices, the \emph{nonuniform recovery guarantees}\footnote{A nonuniform recovery result only states that a fixed pair of sparse signal and sparse corruption can be recovered with high probability using a random draw of the matrix. Sometimes, the signs of the non-zero coefficients of the sparse vector (and corruption) can be chosen at random to further simplify arguments. Uniform recovery is stronger than nonuniform recovery. (see \cite[Chapter 9.2]{foucart2013mathematical}\cite[Section 3.1]{rauhut2010compressive} for more details.)} can be proved for randomly subsampled orthonormal matrix \cite{Nguyen2013exact} and its generalized model - bounded orthonormal systems\footnote{See \cite{candes2011probabilistic} for the construction of the generalized model.} \cite{li2013compressed}. Very recently, the uniform recovery guarantee for bounded orthonormal systems is shown in \cite{adcock2017compressed}.

In this paper, we prove the uniform recovery guarantee for two different corrupted sensing models. In the first model, the measurement matrix is based on randomly modulated unit-norm frames \cite{zhang2014modulated} and the corruption is sparse on the identity basis. It is noted that the measurement matrix in the first model does not consist of a random subsampling operator, e.g., the partial random circulant matrix \cite{krahmer2014suprema}. For the second model, we consider
\begin{align}
\y = \A\xs + \H\zs + \w,
\end{align}
where $\A$ represents a randomly subsampled orthonormal matrix, and the corruption $\H\zs$ is assumed to be sparse on certain bounded domain (e.g., a discrete Fourier transform (DFT) matrix). Our results imply that many structured sensing matrices can be employed in the corrupted sensing model to ensure the exact and stable recovery of both $\xs$ and $\zs$, even when the sparsity of the corruption is up to a constant fraction of the total number of measurements. Thanks to the uniform recovery guarantee, our results can address the adversarial setting, which means that exact and stable recovery is still guaranteed even when $\xs$, $\zs$ and $\w$ are selected given knowledge of the sensing matrix $\A$. In addition, our analysis results are also applicable to demonstrate the recovery guarantee when the corrupted sensing problem is solved via nonconvex optimization.

\subsection{Potential Applications}\label{sec: potential apps}
The problem of recovering sparse signal $\xs$ and sparse corruption $\zs$ from the measurement vector $\y$ arise from many applications, where the compressed measurements may be corrupted by impulse noise.

For example, in a sensor network, each sensor node measures the same signal $\xs$ independently before sending the outcome to the center hub for analysis. In this setting, each sensor makes the measurement $\langle\a_i,\xs\rangle$, and the resultant measurement vector is $\A\xs$ by arranging each $\a_i$ as the rows of $\A$ \cite{haupt2008compressed,Nguyen2013exact}. However, in practice, some sensor readings can be anomalous from the rest. These outliers could be caused by individually malfunctioned sensors, or due to some unusual phenomena or event happening in certain areas of the network \cite{franke2009orden}\cite{zhang2010outlier}. This anomaly effect can be modeled by a sparse vector $\zs$. Mathematically, we have $\y=\A\xs+\zs+\w$, where $\zs$ represents the outlier regions and $\w$ stands for possible small noise in the data transmission. Our results make it possible to recover both the underlying signal and detect the outlier regions simultaneously, which could be very useful for network monitoring.

Another application of sparse signal recovery from sparsely corrupted measurements is  error correction in joint source-channel coding. In \cite{charbi2010compressive,li2010systematic,Nguyen2013exact}, compressed sensing has been exploited as a joint source-channel coding strategy for its efficient encoding and robust error correcting performance. For a signal $\f$ that is sparse in the domain $\PPsi$, i.e., $\f=\PPsi\xs$, it can be encoded by a linear projection $\y=\PPhi\f=\A\xs$ with $\A=\PPhi\PPsi$. Existing works have investigated the situations where the encoded signal $\y$ is sent through either an erasure channel \cite{charbi2010compressive} or a gross error channel \cite{li2010systematic,Nguyen2013exact}. Our results can not only be applied in these scenarios, but also provide a new design on the encoding matrix with uniform recovery guarantee.

In some scenarios, the measurement noise may be sparse or compressible in some sparsifying basis. One example is the recovery of video or audio signal that are corrupted by narrow-band interference (NBI) due to improper designed equipment \cite{laska2009exact,studer2012recovery}. Electric hum as a typical impairment is sparse in the Fourier basis. Another example is the application of compressed sensing to reduce the number of samples in convolution systems with deterministic sequences (e.g., m-sequence, Golay sequence). Such convolution systems are widely used in communications, ultrasound and radar \cite{popovic1991synthesis, davis1999peak}. In practice, the measurements may be affected by frequency domain interference or multi-tone jamming \cite{levitt1985fh}. For instance, in CS-based OFDM channel estimation \cite{berger2010sparse,haupt2010toeplitz,meng2012compressive,li2013convolutional}, suppose $\x$ is the channel response and that the pilot sequence $\g$ is constructed from Golay sequences, the time-domain received signal can be represented as \cite{li2013convolutional}, $\y=\sqrt{\frac{n}{m}}\mR_{\Omega'}\F^*\diag(\g)\F\x + \w$, where $\mR_{\Omega'}$ is a random subsampling operator and $\F$ denotes the DFT matrix. The recovery performance can be guaranteed by noticing that the sensing model is a subsampled version of the orthonormal matrix $\F^*\diag(\g)\F$. However, in OFDM-based powerline communications, the NBI due to intended or unintended narrow-band signals can severely contaminate the transmitted OFDM signal. The time-domain NBI vector is sparse in the Fourier basis \cite{umehara2006performance,gomaa2011sparsity}. Our results cover these settings, and therefore, provide a CS-based method to jointly estimate the signal of interest and the NBI.

\subsection{Notations and Organization of the paper}
For an $n$-element vector $\a$, we denote by $a_i$, ($i\in[n]=\{0,...,n-1\}$), the $i$-th element of this vector. We represent a sequence of vectors by $\a_0,...,\a_{n-1}$ and a column vector with $q$ ones by $\one_q$. The sparsity of a vector can be measured by its best $s$-term approximation error,
\begin{align*}
\sigma_s(\a)_p = \inf_{\|\tilde{\a}\|_0\leq s} \|\a-\tilde{\a}\|_p,
\end{align*}
where $\|\cdot\|_p$ is the standard $l_p$ norm on vectors. For a matrix $\A$, $A_{jk}$ denotes the element on its $j$-th row and $k$-th column. The vector obtained by taking the $j$-th row ($k$-th column) of $\A$ is represented by $\A_{(j,:)}$ ($\A_{(:,k)}$). We denote by $\A_0,...,\A_{n-1}$ a sequence of matrices. $\A^{-1}$ and $\A^*$ represent the inverse and the conjugate transpose of $\A$. The Frobenius norm and the operator norm of matrix $\A$ are denoted by $\|\A\|_F=\sqrt{\text{tr}(\A^*\A)}$ and $\|\A\|_{2\rightarrow2}=\sup_{\xnorm=1}\|\A\x\|_2$ respectively. We write $A\lesssim B$ if there is an absolute constant $c$ such that $A\leq cB$. We denote $A\sim B$ if $c_1A\leq B\leq c_2A$ for absolute constants $c_1$ and $c_2$.

The coherence $\mu(\A)$ of an $\tn\times n$ matrix $\A$ describes the maximum magnitude of the elements of $\A$, i.e., $\mu(\A)=\max_{\substack{1\leq j\leq\tn \\ 1\leq k\leq n}} |\A_{jk}|$. For a unitary matrix $\PPsi\in\R^{n\times n}$, we have $\frac{1}{\sqrt{n}}\leq \mu(\PPsi)\leq 1$.

The rest of the paper is organized as follows. We start by reviewing some key notions and results in compressed sensing in Section \ref{sec: preliminaries}. In Section \ref{sec: main}, we prove the uniform recovery guarantee for two classes of structured random matrices. In Section \ref{sec: simulations}, we conduct a series of simulations to reinforce our theoretical results. Conclusion is given in Section \ref{sec: conclusion}. We defer most of the proofs to the Appendices.
\section{Preliminaries}\label{sec: preliminaries}
\subsection{RIP and structured sensing matrices}
The restricted isometry property (RIP) is a sufficient condition that guarantees uniform and stable recovery of all $s$-sparse vectors via nonlinear optimization (e.g. $l_1$-minimization). For a matrix $\A\in\R^{m\times n}$ and $s<n$, the restricted isometry constant $\ric$ is defined as the smallest number such that
\begin{align*}
(1-\ric)\|\x\|_2^2\leq\|\A\x\|_2^2\leq(1+\ric)\|\x\|_2^2,
\end{align*}
holds for all $s$-sparse vectors $\x$. Alternatively, the restricted isometry constant of $\A$ can be written as
\begin{align}
\ric=\sup_{\x\in\sset}\left|\|\A\x\|_2^2-\|\x\|_2^2\right|, \label{eqn: ric dfn}
\end{align}
where $\sset=\{\x\in\R^n: \|\x\|_2\leq 1, \|\x\|_0\leq s\}$.

Among the many structured sensing matrices that satisfy the RIP, two classes have been found to be applicable in various scenarios. One is the randomly subsampled orthonormal systems \cite{rauhut2010compressive}, which encompass structured sensing matrices like partial random Fourier \cite{EJC06nearoptimal}, convolutional CS \cite{romberg2009compressive,li2013convolutional} and spread spectrum \cite{puy2012universal}. The other is the UDB framework which consists of a unit-norm tight frame (UTF), a random diagonal matrix and a bounded column-wise orthonormal matrix \cite{zhang2014modulated}. Popular sensing matrices under this framework include partial random circulant matrices \cite{krahmer2014suprema}, random demodulation \cite{tropp2010beyond}, random probing \cite{romberg2010sparse} and compressive multiplexing \cite{slavinsky2011compressive}.

\subsection{Recovery Condition}
We review the definition of generalized RIP, which is useful to establish robustness and stability of the optimization algorithm.

\begin{dfn}\cite[Definition 2.1]{li2013compressed}
For any matrix $\TTheta\in\R^{r\times(n+m)}$, it has the $(s,k)$-RIP with constant $\delta_{s,k}$ if $\delta_{s,k}$ is the smallest value of $\delta$ such that
\begin{align}
(1-\delta)(\|\x\|_2^2+\|\z\|_2^2)\leq\left\|\TTheta\begin{bmatrix}\x\\ \z\end{bmatrix}\right\|_2^2\leq(1+\delta)(\|\x\|_2^2+\|\z\|_2^2)
\end{align}
holds for any $\x\in\R^{n}$ with $\|\x\|_0\leq s$ and any $\z\in\R^{m}$ with $\|\z\|_{0}\leq k$.
\end{dfn}

Here, the generalized RIP is termed as the $(s,k)$-RIP for convenience. We note that the $(s,k)$-RIP is more stringent than the conventional RIP. In other words, the fact that a sensing matrix $\A$ satisfies the RIP does not mean that the associated matrix $\TTheta=[\A,\ \I]$ would satisfy the $(s,k)$-RIP. The recovery performance of the penalized optimization \eqref{eqn: penalized rec} can be guaranteed by the following result.

\begin{theorem}\label{thm: recovery condition}\cite[Theorem 3.7]{adcock2017compressed}
Suppose $\y=\A\xs+\zs+\w$ and $\TTheta=[\A,\ \I]\in\R^{m\times(n+m)}$ has the $(2s,2k)$-RIP constant $\delta_{2s,2k}$ satisfying
\begin{align*}
\delta_{2s,2k}< \frac{1}{\sqrt{1+\left(\frac{1}{2\sqrt{2}}+\sqrt{\eta}\right)^2}}
\end{align*}
with $\eta=\frac{s+\lambda^2 k}{\min \{s, \lambda^2 k\}}$. Then for $\xs\in\R^n$, $\zs\in\R^m$, and $\w\in\R^m$ with $\|\w\|_2\leq\varepsilon$, the solution $(\xr,\zr)$ to the penalized optimization problem \eqref{eqn: penalized rec} satisfies
\begin{align*}
\|\xr-\xs\|_1+\|\zr-\zs\|_1\ & \leq c_1(\sigma_s(\x)_1+\lambda\sigma_k(\z))\\
&\quad + c_2\sqrt{s+\lambda^2k}\varepsilon \\
\|\xr-\xs\|_2+\|\zr-\zs\|_2\ & \leq c_3 \left(1+\eta^{1/4}\right)\left(\frac{\sigma_s(\x)_1}{\sqrt{s}}+\frac{\sigma_k(\z)_1}{\sqrt{k}}\right) \\
&\quad + c_4\left(1+\eta^{1/4}\right)\varepsilon,
\end{align*}
where the constants $c_1$, $c_2$, $c_3$, $c_4$ depend on $\delta_{2s,2k}$ only.
\end{theorem}
We note that similar theorem has been proven in \cite{li2013compressed} when both the signal and corruption are vectors with exact sparsity. The above result not only relaxes the requirement on the $(2s,2k)$-RIP constant, but also guarantees stable recovery of inexactly sparse signals and corruptions. Therefore, for either sparse or compressible signals and corruptions, the key to establish the recovery guarantee for a sensing matrix is to prove the $(s,k)$-RIP.

\section{Main Results}\label{sec: main}
In this section, we prove the $(s,k)$-RIP for two classes of structured sensing matrices. This result can then be combined with Theorem \ref{thm: recovery condition} to prove the recovery guarantee. In addition, the extension to the recovery via nonconvex optimization is presented. Last but not least, we compare the main theorems to existing literature where relevant.

\subsection{Randomly modulated unit-norm tight frames}
We prove the uniform recovery guarantees for the class of structured sensing matrices that can be written as $\A=\U\D\tB$, where $\U\in\R^{m\times \tn}$ is a UTF with $\mu(\U)\sim 1/\sqrt{m}$, $\D=\diag(\bxi)$ is a diagonal matrix with $\bxi$ being a length-$\tn$ random vector with independent, zero-mean, unit-variance, and $L$-subgaussian entries, and $\tB\in\R^{\tn\times n}$, $\tn\geq n$, represents a column-wise orthonormal matrix, i.e. $\tB^*\tB=\I$.

The following result presents a bound on the required number of measurements $m$ such that the corresponding matrix $\TTheta$ has the $(s,k)$-RIP constant satisfying $\delta_{s,k}\leq\delta$ for any $\delta\in (0,1)$.
\begin{theorem}\label{thm: the skRIP of udb}
Suppose $\y=\A\xs+\zs+\w$ with $\TTheta=[\A,\ \I]\in\R^{m\times(n+m)}$, $\A=\U\D\tB$ and $\mu(\U)\sim 1/\sqrt{m}$. If, for $\delta\in (0,1)$,
\begin{align*}
m&\geq c_5\delta^{-2}s\tn\mu^2(\tB)\log^2s\log^2\tn,\\
m&\geq c_6\delta^{-2}k\log^2 k\log^2\tn,
\end{align*}
where $c_5$ and $c_6$ are some absolute constants, then with probability at least $1-2\tn^{-\log^2 s\log \tn}$, the $(s,k)$-RIP constant of $\TTheta$ satisfies $\delta_{s,k}\leq\delta$.
\end{theorem}

\begin{proof}
The $(s,k)$-RIP constant $\skripc$ can be equivalently expressed as
\begin{align}\label{eqn: ripc equivalent expression}
\skripc=\sup_{(\x,\z)\in\mT}\left|\left\|\TTheta\begin{bmatrix}\x\\ \z\end{bmatrix}\right\|_2^2-\|\x\|_2^2-\|\z\|_2^2\right|,
\end{align}
where $\mT:=\{(\x,\z):\|\x\|_2^2+\|\z\|_2^2=1,\|\x\|_0\leq s,\|\z\|_{0}\leq k, \x\in\R^{n}, \z\in\R^{m}\}$. With $\TTheta=[\A,\ \I]$, the RIP-constant can be further reduced to
\begin{align}\label{eqn: ripc decomposition}
\skripc&=\sup_{(\x,\z)\in\mT}\left|\|\A\x\|_2^2+\|\z\|_2^2+2\langle\A\x,\z\rangle-\|\x\|_2^2-\|\z\|_2^2\right|\nonumber\\
      &\leq \underbrace{\sup_{(\x,\z)\in\mT}\left|\|\A\x\|_2^2-\|\x\|_2^2\right|}_{\delta_1}+\underbrace{2\sup_{(\x,\z)\in\mT}|\langle\A\x,\z\rangle|}_{\delta_2}
\end{align}
Our aim is to derive bounds on the number of measurements $m$ such that for any $\delta\in(0,1)$ the RIP-constant $\skripc$ is upper bounded by $\delta$. We have
\begin{align}
\delta_1&\leq \sup_{\x\in\sset}\left|\|\A\x\|_2^2-\|\x\|_2^2\right|
\end{align}
with $\sup_{\x\in\sset}\left|\|\A\x\|_2^2-\|\x\|_2^2\right|$ being the restricted isometry constant in the standard RIP definition \eqref{eqn: ric dfn}. Then, by \cite[Theorem III.2]{zhang2014modulated}, we reach the following result.

Suppose, for any $\delta\in (0,1)$, 
\begin{align*}
m\geq 4c_1\delta^{-2}s\tn\mu^2(\tB)(\log^2s\log^2\tn),
\end{align*}
then $\delta_1\leq\delta/2$ holds with probability at least $1-\tn^{-(\log \tn)(\log s)^2}$.

Therefore, proof of the $(s,k)$-RIP is reduced to bounding the inner product term $\delta_2$.
\begin{align}
\delta_2&=2\sup_{(\x,\z)\in\mT}|\langle\U\D\tB\x,\z\rangle|=2\sup_{(\x,\z)\in\mT}|\z^{*}\U\D\tB\x|\nonumber\\
        &=2\sup_{(\x,\z)\in\mT}|\z^{*}\U\diag(\tB\x)\bxi|=2\sup_{\v\in\vset}|\langle\v,\bxi\rangle|, \label{eqn: detla 2}
\end{align}
where $\v=(\z^{*}\U\diag(\tB\x))^*$, and 
\begin{align}\label{eqn: set of vectors}
\vset:=\{\v:\|\x\|_2^2+\|\z\|_2^2=1,\|\x\|_0\leq s,\|\z\|_{0}\leq k\}.
\end{align}
The following lemma is proved in Appendix \ref{appdix: proof of useful lemma}.
\begin{lemma}\label{lem: useful lemma bound the inner product}
Suppose $\bxi$ is a length-$\tn$ random vector with independent, zero-mean, unit-variance, and $L$-subgaussian entries. For any $\delta\in(0,1)$, if
\begin{align*}
m&\geq c_5\delta^{-2}s\tn\mu^2(\tB)\log^2 s\log^2 \tn\\
m&\geq c_6\delta^{-2}k\log^2 k\log^2 \tn,
\end{align*}
then $\sup_{\v\in\vset} |\langle\v,\bxi\rangle|\leq \delta/2$ holds with probability exceeding $1-\exp(-\log^2 s\log^2 \tn)$, where $c_5$ and $c_6$ are some constants depending only on $L$.
\end{lemma}

Combining \eqref{eqn: detla 2} with Lemma \ref{lem: useful lemma bound the inner product}, we have, for any $\tau>0$, $\delta_2\leq c\delta$ holds with probability exceeding $1-\exp(-\log^2 s\log^2 \tn)$ for some constant $c$.

Finally, Theorem \ref{thm: the skRIP of udb} can be obtained by combining the above results. Suppose, for any $\delta\in (0,1)$, $m\geq c_5\delta^{-2}s\tn\mu^2(\tB)(\log^2s\log^2\tn)$, $m\geq c_6\delta^{-2}k\log^2 m\log^2 \tn$ and $\mu(U)\sim1/\sqrt{m}$, then we have $\skripc\leq \delta_1+\delta_2\leq \delta$ with probability exceeding
\begin{align*}
&1-\tn^{-(\log \tn\log^2 s)}-\exp(-c\log^2 s\log^2 \tn)\\
&=1-\tn^{-(\log \tn\log^2 s)}-\tn^{-c\log^2 s\log \tn}\\
&=1-2\tn^{-\log^2 s\log \tn}
\end{align*}
\end{proof}
The uniform recovery guarantee can be obtained by combining Theorem \ref{thm: recovery condition} and \ref{thm: the skRIP of udb}.

A few remarks are in order. First, when $\tB$ is a bounded column-wise orthonormal matrix, i.e., $\mu(\tB)\sim 1/\sqrt{\tn}$, the bound on the sparsity of $\xs$ can be relaxed to $\|\xs\|_0\leq Cm/(\log^2\tn\log^2 m)$. The sparsity $\|\zs\|_{0}$ is always a constant fraction of the total number of measurements $m$ regardless the magnitude of the coherence $\mu(\tB)$. When $\|\w\|_2=0$, Theorem \ref{thm: the skRIP of udb} implies that a sparse signal can be exactly recovered by tractable $l_1$ minimization even if some parts of the measurements are arbitrarily corrupted.

Second, the proposed class of structured sensing matrices is equivalent to the UDB framework \cite{zhang2014modulated} but with an additional requirement of $\mu(\U)\sim 1/\sqrt{m}$. The UDB framework has been proved to support uniform recovery guarantees for conventional CS problem, while with the extra condition it is now shown to provide uniform recovery guarantees for the CS with sparse corruptions problem. Theorem \ref{thm: the skRIP of udb} holds for many existing and new structured sensing matrices as long as they can be decomposed into $\A=\U\D\tB$.

One application of the UDB framework is to simplify the mask design in double random phase encoding (DRPE) for optical image encryption. Consider an image $\f$ that is sparse in the domain $\PPsi$, i.e., $\f=\PPsi\xs$, DRPE is based on random masks placed in the input and Fourier planes of the optical system \cite{rivenson2010single,li2015compressiveoptical}
. Mathematically, the measurements can be written as $\y=\sqrt{\frac{n}{m}}\mR_{\OOmega}\F^*\LLambda_1\F\LLambda_2\f+\w$, where $\mR_{\OOmega}: \C^n \rightarrow\C^m$ represents an arbitrary/deterministic subsampling operator with $\OOmega$ being the set of selected row indices, $\LLambda_1$ and $\LLambda_2$ are random diagonal matrices.
By the UDB framework, the random diagonal matrix $\LLambda_2$ can be replaced by a deterministic diagonal matrix constructed from a Golay sequence $\g$. The reason is that the measurement model $\sqrt{\frac{n}{m}}\mR_{\OOmega}\F^*\LLambda_1\F\diag(\g)\PPsi\xs$ can be decomposed into a UTF $\sqrt{\frac{n}{m}}\mR_{\OOmega}\F^*$, a random diagonal matrix $\LLambda_1$, and a orthonormal matrix $\F\diag(\g)\PPsi$ whose coherence is proven to be bounded for many orthonormal transforms $\PPsi$, e.g., DCT, Haar wavelet \cite[Lemma IV.2]{zhang2014modulated}. When the measurements are corrupted by impulse noise due to detector plane impairment, our theorem above provides a recovery guarantee on the image.

Furthermore, the UDB framework emcompasses some popular structured sensing matrices, e.g., partial random circulant matrices \cite{krahmer2014suprema} and random probing \cite{romberg2010sparse}. To elaborate, consider the partial random circulant matrices
\begin{align*}
\A=\frac{1}{\sqrt{m}}\mR_{\OOmega}\mC_{\eps}
\end{align*}
where $\mC_{\eps}$ denotes the circulant matrix generated from $\eps$. Suppose $\eps=\F^*\bxi$, where $\F$ is a normalized DFT matrix and $\bxi$ is a length-$n$ random vector with independent, zero-mean, unit-variance, and sub-Gaussian entries. Let $\D=\diag(\bxi)$, we have $\A=\sqrt{\frac{n}{m}}\mR_{\OOmega}\F^*\D\F$. It can be observed that $\U=\sqrt{\frac{n}{m}}\mR_{\OOmega}\F^*$ is a UTF and $\B=\F$ is a unitary matrix. Hence, Theorem \ref{thm: the skRIP of udb} implies that any sparse signal $\x$ and sparse corruption $\z$ can be faithfully recovered from the measurement model $\y=\frac{1}{\sqrt{m}}\mR_{\OOmega}\mC_{\eps}\xs+\zs+\w$ by the penalized recovery algorithm. The sparse recovery from partial random circulant measurements can be applied in many common deconvolution tasks, such as radar \cite{herman2009high} and coded aperture imaging \cite{marcia2008fast}. In practice, where the measurements can be corrupted by impulsive noise due to bit errors in transmission, faulty memory locations, and buffer overflow \cite{pham2012improved}, our theorem guarantees the recovery of both the signal of interest and the corruption.

In some situations, the proposed framework can still provide reliable recovery guarantee even if the corruption is sparse in some basis. Suppose the corruption is sparse under some fixed and known orthonormal transformation $\H$, i.e. $\H^{*}\H=\I$. We consider the measurement model
\begin{align}
\y=\A\xs+\H\zs+\w.
\end{align}
It is clear that this setting can be reduced to
\begin{align}
\H^{*}\y=\H^{*}\A\xs+\zs+\H^{*}\w.
\end{align}
Notice that $\H^{*}\A=\H^{*}\U\D\tB:=\widehat{\U}\D\tB$, where $\widehat{\U}=\H^{*}\U$ is still a UTF due to the orthogonality of $\H$. Therefore, if $\mu(\widehat{\U})\sim 1/\sqrt{m}$, Theorem \ref{thm: the skRIP of udb} still holds in this measurement model.

\subsection{Randomly sub-sampled orthonormal system}
Next, we consider the corrupted sensing measurement model for randomly sub-sampled orthonormal system. We prove the uniform recovery guarantee for such matrices provided that the corruption is sparse on certain sparsifying domain. Suppose $\llambda\in\mathbb{R}^n$ is a random Bernoulli vector with i.i.d. entries such that $\Prob(\lambda_i=1)=\frac{m}{n}$ $\forall i\in[n]$ and $\Omega'=\{i: \lambda_i=1\}$ with $|\Omega'|=M$, the random sampling operator $\mR_{\Omega'}\in\mathbb{R}^{M\times n}$ is a collection of the $i$-th row of an $n$-dimensional identity matrix for all $i\in\Omega'$. Here, $M$ is random with mean value $m$. The observation model is
\begin{align}\label{eqn: subfft model}
\y=\A\xs+\H\zs+\w,
\end{align}
where $\A=\sqrt{\frac{n}{M}}\mR_{\Omega'}\G$, $\G\in\C^{n\times n}$ is an orthonormal basis and $\H\in\C^{M\times M}$ is a unitary matrix with $\mu(\H)\sim 1/\sqrt{M}$.

From our analysis in previous subsection, the uniform recovery performance can be guaranteed as long as the associated matrix $\TTheta$ satisfies the $(s,k)$-RIP. Since the matrix $\A$ satisfies the standard RIP, the problem of proving the $(s,k)$-RIP is again reduced to bounding the inner product term $\sup_{(\x,\z)\in\mT}|\langle\A\x,\H\z\rangle|$. Detail proof of the following result is given in Appendix \ref{app: proof of skRIP of rdn unitary mtx}.

\begin{theorem}\label{thm: skRIP of rdn unitary mtx}
Suppose $\y=\A\xs+\H\zs+\w$ with $\TTheta=[\A, \ \H]\in\R^{M\times (n+M)}$, $\A=\sqrt{\frac{n}{M}}\mR_{\Omega'}\G$ and $\mu(\H)\sim 1/\sqrt{M}$. If, for $\delta\in (0,1)$,
\begin{align*}
m&\geq \max (c_7\delta^{-2}s n\mu^2(\G)\log^2s\log^2 n, c_8\delta^2s\log^4 n, 2c_9\log n),\\
m&\geq c_{10}\delta^{-2}k n\mu^2(\G)\log^2 k\log^2 n,\\
m&\leq c_{11}\delta^2 n,
\end{align*}
where $\{c_i\}_{i=7,...,11}$ are constants, then with probability at least $1-2\tn^{-\log^2 s\log n}-n^{-c_9}$, the $(s,k)$-RIP constant of $\TTheta$ satisfies $\delta_{s,k}\leq\delta$.
\end{theorem}

When $\G$ is a bounded orthonormal basis, i.e., $\mu(\G)\sim 1/\sqrt{n}$, the bound on the sparsity of $\xs$ can be relaxed to $m\geq\bigo(\delta^{-2}s \log^2s\log^2 n, \delta^{-2}k \log^2 k\log^2 n)$, which implies that a sparse signal can be exactly recovered by tractable $l_1$ minimization even if the measurements are affected by corruption sparse on some bounded domain. A bounded orthonormal basis can include the Fourier transform or the Hadamard transform. In addition, in CS-based OFDM where the pilot is generated from a Golay sequence and a random subsampler is employed at the receiver (Section \ref{sec: potential apps}), the effective orthonormal basis is also bounded, i.e., $\mu(\F^*\diag(\g)\F)\sim 1/\sqrt{n}$ \cite{li2013convolutional}.

\subsection{Nonconvex optimization}
We have shown the $(s,k)$-RIP for two popular classes of structured sensing matrices, and proven the performance guarantee for the recovery of the sparse signal and corruption via the $l_1$-norm minimization algorithm \eqref{eqn: penalized rec}. However, our $(s,k)$-RIP analysis on the structured sensing matrices is also applicable to proving the recovery guarantee for nonconvex optimization. Consider the following formulation of the problem
\begin{align}
\y=\A\xs+\H\zs,
\end{align}
It was demonstrated in \cite{filipovic2014reconstruction} that the unique minimizer of the $l_p$ minimization problem ($0<p<1$)
\begin{align}
\min_{\x,\z}\|\x\|_p^p+\nu\|\z\|_p^p\quad\quad \mathrm{s.t.} \quad\quad \A\x+\H\z=\y,
\end{align}
is exactly the pair $(\xs,\zs)$ if the combined matrix $[\A, \H]$ satisfies the $(s,k)$-RIP, where $\nu$ is the regularization parameter. In addition, the $l_p$ minimization approach still provides stable recovery even when there is additional dense noise as long as the $(s,k)$-RIP holds \cite{saab2008stable,filipovic2014reconstruction}. The $l_p$ minimization problem can be numerically solved via an iteratively reweighted least squares (IRLS) method \cite{chartrand2008iteratively}. However, \cite{filipovic2014reconstruction} only considers the sensing model with $\A$ being random Gaussian matrices and $\H$ being an identity matrix. With our $(s,k)$-RIP analysis, many structured sensing matrices can be employed to provide exact/stable recovery performance for this corrupted sensing problem via $l_p$ minimization.

\subsection{Comparison with related literature}
In this part, we compare our main results with related literature.
\subsubsection{Sparse signal, sparse corruption}
\cite{li2013compressed} proved that sensing matrices with independent Gaussian entries provide uniform recovery guarantee for corrupted CS by solving \eqref{eqn: penalized rec} for all vectors $\xs$ and $\zs$ satisfying $\|\x\|_0\leq \alpha m/(\log (n/m)+1)$ and $\|\zs\|_0\leq \alpha m$. The difference is that our theorems come with a tighter requirement on the sparsity of $\xs$ and the sparsity of $\zs$, which is a compensation on the employment of structured measurements.

\cite{li2013compressed} also proved the recovery guarantee for structured sensing matrices that belong to the framework proposed in \cite{candes2011probabilistic}. Here, faithful recovery is guaranteed provided that $\|\x\|_0\leq \alpha m/(\mu^2\log^2 n)$ and $\|\zs\|_0\leq \beta m/\mu^2$, where $\mu$ is the coherence of the sensing matrix. \cite{Nguyen2013exact} considered the corrupted CS with sensing matrices that are randomly subsampled orthonormal matrix, and proved similar results. It is noted that the requirements on the sparsity of $\xs$ in these works seem less strict than that in our results. However, in both \cite{li2013compressed} and \cite{Nguyen2013exact}, performance guarantees of their structured sensing matrices rely on the assumption that the support set of $\xs$ or $\zs$ is fixed and the signs of the signal are independently and equally likely to be $1$ or $-1$ \cite[Section 1.3.2]{li2013compressed}\cite[Section II.B]{Nguyen2013exact} (i.e. a nonuniform recovery guarantee). While in our theorem, two classes of structured sensing matrices (including randomly subsampled orthgonal matrix) are shown to provide uniform recovery guarantee for corrupted CS.

We note that recently the uniform recovery guarantee for bounded orthonormal systems is proven in \cite{adcock2017compressed}. The bounded orthonormal systems is more general than the random subsampled orthonormal matrix considered in our second class. However, the corruption models are different: the corruption vector in \cite{adcock2017compressed} is sparse in time domain, whereas our theorem considers corruption in sparsifying domain with $\mu(\H)\sim 1/\sqrt{M}$. Due to this difference in the corruption model, the techniques used to prove the $(s,k)$-RIP (specifically, bound the inner product $\sup_{(\x,\z)\in\mT}|\langle\A\x,\H\z\rangle|$) are essentially different.

\subsubsection{Structured signal, structured corruption}
In a recent work \cite{Foygel2014corrupted}, sensing with random Gaussian measurements for general structured signals and corruptions (including sparse vectors, low rank matrix, sign vectors and etc) has been proven. However, our study departs from it in the following aspects: \cite{Foygel2014corrupted} proved a nonuniform recovery guarantee for the recovering of sparse signals from sparse corruptions and dense noise. In our paper, we established a uniform recovery guarantee for the corresponding problem. Moreover, \cite{Foygel2014corrupted} considered random Gaussian matrices, while we propose structured sensing matrices.

We have shown that a large class of structured sensing matrices can provide faithful recovery for the sparse sensing with sparse corruption. Whether such structured measurements can be applied in a general corrupted sensing problem (e.g. structured signal with structured corruption) is still open. Extension of our measurement framework to the general corrupted sensing problem is interesting for further study.

Other works related to the recovery of signals from corrupted measurements include \cite{chen2013robust,studer2014stable,pope2013probabilistic,wright2009robust,laska2009exact,laska2011democracy,chiconvex,fernandez2016demixing}. However, their models are different from the one in our paper.

\begin{figure*}[!t]
\centering
\subfloat[Gaussian Setting]{\includegraphics[width=3.5in]{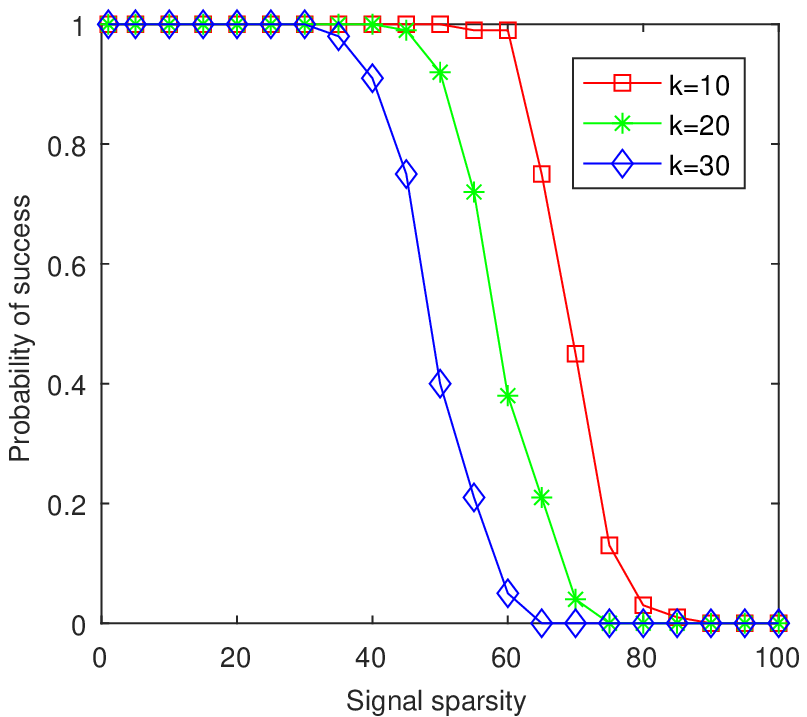}\label{fig: Eg1a}}
\hfil
\subfloat[Flat Setting]{\includegraphics[width=3.5in]{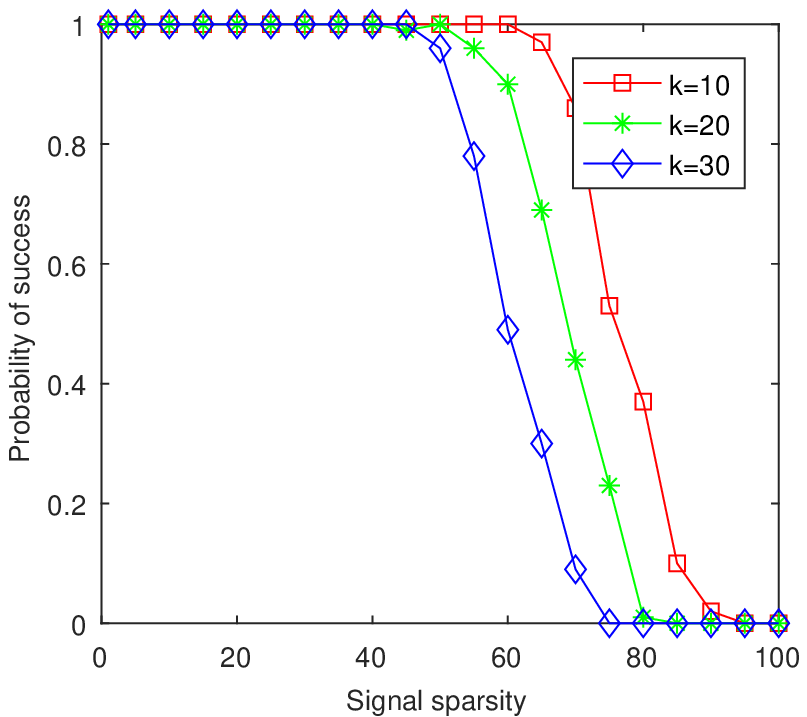}\label{fig: Eg1b}}
\caption{Probability of success as a function of the signal sparsity $s$ using penalized recovery with signal dimension $n=512$, number of measurements $m=256$, and the corruption sparsity $k=\{10, 20, 30\}$ for Mtx-I.}
\label{fig: Eg1}
\end{figure*}

\begin{figure*}[!t]
\centering
\subfloat[Gaussian Setting]{\includegraphics[width=3.5in]{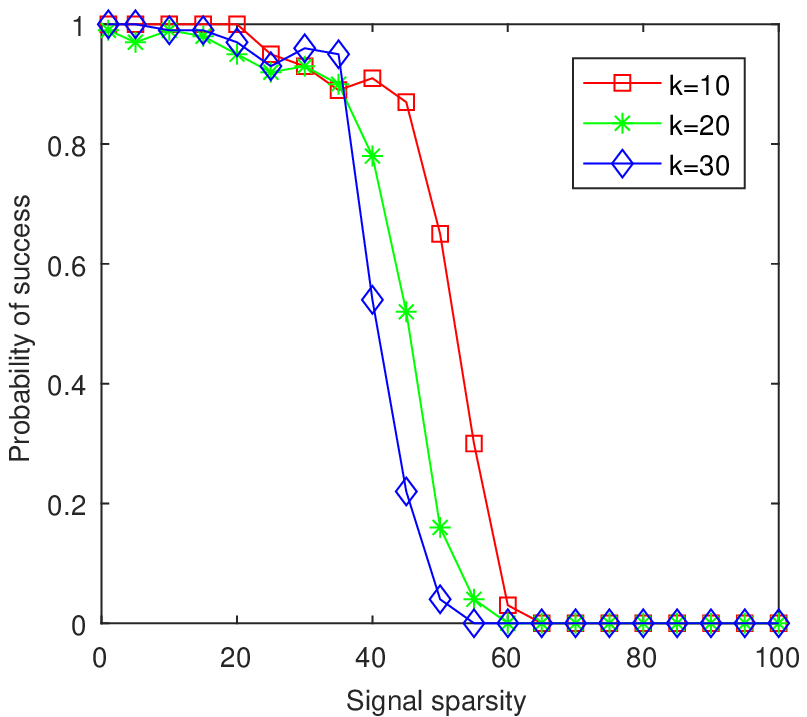}\label{fig: Eg2a}}
\hfil
\subfloat[Flat Setting]{\includegraphics[width=3.5in]{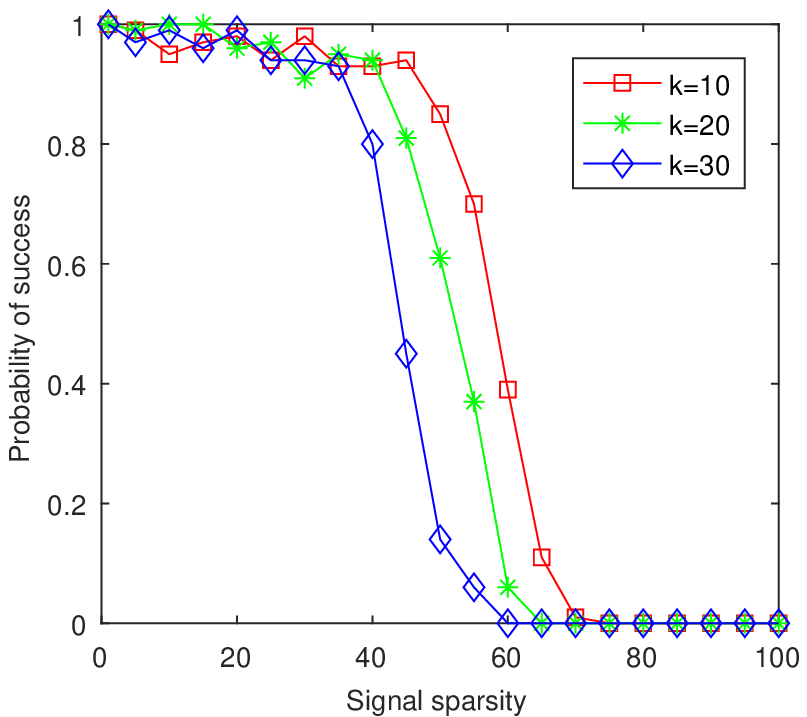}\label{fig: Eg2b}}
\caption{Probability of success as a function of the signal sparsity $s$ using penalized recovery with signal dimension $n=512$, number of measurements $m=256$, and the corruption sparsity $k=\{10, 20, 30\}$ for Mtx-II.}
\label{fig: Eg2}
\end{figure*}

\begin{remark}
\emph{We note that the value of the regularization parameter can be chosen as $\lambda=\sqrt{s/k}$. In practice, when no a priori knowledge on the sparsity levels of the signal and the corruption is available, $\lambda$ can usually be taken by cross validation. On the other hand, if it is known a priori that the corruption (the signal) is very sparse, one can increase (decrease) the value of $\lambda$ to improve the overall recovery performance. Similar discussion on the theoretical and practical settings of the regularization parameter has also been noted in \cite[Section 1.3.3]{li2013compressed}, \cite[Section II.E, Section VII]{Nguyen2013exact}, \cite[Section III.B]{Foygel2014corrupted}. In addition, an iteratively reweighted $l_1$ minimization method can be used to adaptively improve the setting of $\lambda$ in practice \cite{adcock2017compressed}.}
\end{remark}

\section{Numerical Simulations}\label{sec: simulations}
In this section, we verify and reinforce the theoretical results of Section \ref{sec: main} with a series of simulations. We present experiments to test the recovery performance of the penalized recovery algorithm for the proposed structured sensing matrices. In each experiment, we used the CVX Matlab package \cite{cvx,gb08} to specify and solve the convex recovery programs.

Two different ways of generating sparse vectors are considered:
\begin{itemize}
  \item Gaussian setting: the nonzero entries are drawn from a Gaussian distribution and their locations are chosen uniformly at random,
  \item Flat setting: the magnitudes of nonzero entries are $1$ and their locations are chosen uniformly at random.
\end{itemize}

\subsection{Penalized Recovery}\label{subsec: simu Eg1}
This experiment is to investigate the empirical recovery performance of the penalized recovery algorithm \eqref{eqn: penalized rec} when the dense noise is zero. Here, the sensing matrix (Mtx-I) $\A=\U\D\B$ of size $m\times n$ with $m=256$ and $n=512$ is constructed as below.
\begin{enumerate}
  \item Arbitrarily select $m=256$ rows from a $512\times 512$ Hadamard matrix to form a new matrix, which is then normalized by $1/\sqrt{m}$ to form the UTF $\U$.
  \item The diagonal entries of the diagonal matrix $\D$ are i.i.d. Bernoulli random variables.
  \item $\B$ is a normalized Hadamard matrix.
\end{enumerate}
We vary the signal sparsity and the corruption sparsity with $s\in[1,100]$ and $k\in\{10,20,30\}$. For each pair of $(s,k)$, we draw a sensing matrix as described above and perform the following experiment $100$ times:
\begin{enumerate}
  \item Generate $\xs$ with sparsity $s$ by the Gaussian setting
  \item Generate $\zs$ with sparsity $k$ by the Gaussian setting
  \item Solve \eqref{eqn: penalized rec} by setting $\lambda=1$
  \item Declare success if\footnote{This criterion indicates that both $\xs$ and $\zs$ have been faithfully recovered.}
  \begin{align*}
  \|\xr-\xs\|_2/\|\xs\|_2+\|\zr-\zs\|_2/\|\zs\|_2<10^{-3}
  \end{align*}
\end{enumerate}
The fraction of successful recovery averaged over the $100$ iterations is presented in Fig. \ref{fig: Eg1a}. To demonstrate the performance for signals and corruptions that do not have i.i.d. signs, the experiment is repeated by generating the sparse vectors $\xs$ and $\zs$ based on the Flat setting as shown in Fig. \ref{fig: Eg1b}. It can be seen that in both scenarios the performance improves as the sparsity of the corruption decreases.

Next, we demonstrate the performance of the penalized recovery algorithm when the sensing matrix is from a randomly subsampled orthonormal matrix. The sensing matrix (Mtx-II) $\A$ is a collection of randomly selected $M=256$ rows from a $512\times 512$ Hadamard matrix, and normalized by $\sqrt{n/M}$. The corruption is $\H\z$, where $\H$ is an $M\times M$ normalized Hadamard matrix. For each pair of $(s,k)$, we repeat the above steps $100$ times to obtain the probability of success (see Fig. \ref{fig: Eg2}). It is noted that the recovery performance of Mtx-I is better than that of Mtx-II. This seems consistent with our theoretical analysis as the random subsampled orthonormal matrix shows more stringent recovery condition than the UDB framework (see Theorem \ref{thm: the skRIP of udb} and \ref{thm: skRIP of rdn unitary mtx}). However, since the $(s,k)$-RIP is a sufficient condition for the recovery guarantee, it may not fully reflect the performance gap between the two classes of structured sensing matrices. Further investigation based on a necessary and sufficient condition for the recovery guarantee of the corrupted CS problem is a difficult, but interesting open question.

\subsection{Stable recovery}
We study the stability of the penalized recovery algorithms when the dense noise is nonzero, i.e., $\varepsilon\neq0$, and compare the performance of structured sensing matrix (Mtx-I) with random Gaussian sensing matrix. In this experiment, the $256$-by-$512$ sensing matrix (Mtx-I) is constructed as in previous subsection. We fix the signal and corruption sparsity levels at $s=10$ and $k=10$ respectively. The dense noise $\w$ consists of i.i.d. Bernoulli entries with amplitude $\varepsilon$. We vary the noise level with $\varepsilon\in[0, 0.1]$, and perform the following experiment $100$ times for each $\varepsilon$:
\begin{enumerate}
  \item Generate $\xs$ with $s=10$ by the Gaussian setting
  \item Generate $\zs$ with $k=10$ by the Gaussian setting
  \item Solve penalized recovery (p-rec) algorithm \eqref{eqn: penalized rec} by setting $\lambda=1$
  \item Record the empirical recovery error $\|\xr-\xs\|_2+\|\zr-\zs\|_2$
\end{enumerate}

\begin{figure}[!t]
\centering
\includegraphics[width=3.5in]{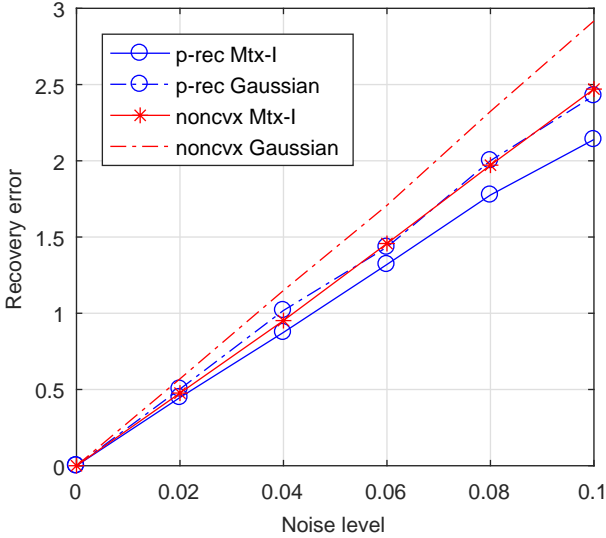}
\caption{Empirical recovery error versus the noise level $\varepsilon$.}
\label{fig: Eg3}
\end{figure}
An average recovery error is then obtained for each $\varepsilon$. Fig. \ref{fig: Eg3} depicts the average error with varying noise levels. The results in Theorems \ref{thm: recovery condition} and \ref{thm: the skRIP of udb} imply that the recovery errors are bounded by the noise level $\varepsilon$ up to some constants. Fig. \ref{fig: Eg3} clearly shows this linear relationship. In addition, we repeat the above experiments with an iteratively reweighted least squares approach \cite{chartrand2008iteratively} using $p=0.5$. As shown in Fig. \ref{fig: Eg3}, the structured sensing matrix is still able to exhibit stable performance by the nonconvex optimization algorithm.

\section{Conclusion}\label{sec: conclusion}
We have studied a generalized CS problem where the measurement vector is corrupted by both sparse noise and dense noise. We have proven that structured random matrices encompassed in the UDB framework or the randomly subsampled orthonormal system can satisfy the sufficient condition, i.e., the $(s,k)$-RIP. These structured matrices can therefore be applied to provide faithful recovery of both the sparse signal and the corruption by the penalized optimization algorithm as well as the nonconvex optimization algorithm. Our simulations have clearly illustrated and reinforced our theoretical results.

\appendices
\section{Proof of Lemma \ref{lem: useful lemma bound the inner product}}\label{appdix: proof of useful lemma}
Throughout the proof in this and the following sections, $C$ and $c$ denote an absolute constant whose values may change from occurrence to occurrence.

A metric space is denoted by $(T,d)$, where $T$ is a set and $d$ is the notion of distance (metric) between elements of the set. For a metric space $(T,d)$, the covering number $N(T,d,u)$ is the minimal number of open balls of radius $u$ needed to cover $(T,d)$. A subset $\overbar{T}$ of $T$ is called a $u$-net of $T$ if every point $\x\in T$ can be approximated to within $u$ by some point $\bar{\x}\in \overbar{T}$, i.e., $d(\x,\bar{\x})\leq u$. The minimal cardinality of $\overbar{T}$ is equivalent to the covering number $N(T,d,u)$. The $p$-th moment (or the $L_p$-norm) of a random variable is denoted by $\|X\|_{L_p}=(\mE|X|^p)^{1/p}$.

We aim to upper bound the variable $\Delta:=\sup_{\v\in\mA_{\v}} |\langle\v,\bxi\rangle|$ which is the supremum of a stochastic process with the index set $\mA_{\v}$. To complete the proof, we require the following important result due to Krahmer et al.:

\begin{theorem}\cite[Theorem 3.5 (a)]{krahmer2014suprema}\label{thm: bounds on chaos processes}
Let $\mset$ be a set of matrices, and let $\bxi$ be a random vector whose entries $\xi_j$ are independent, mean $0$, variance $1$, and $L$-subgaussian random variables. Set
\begin{align*}
d_F(\mset)&=\sup_{\S\in\mset}\|\S\|_{F},\quad d_{\ttwo}(\mset)=\sup_{\S\in\mset}\|\S\|_{\ttwo},\\
N_{\mset}(\bxi)&:=\sup_{\S\in\mset}\|\S\bxi\|_2,\quad E=\gamma_2(\mset,\|\cdot\|_{\ttwo})+d_F(\mset).
\end{align*}
Then, for every $p\geq 1$,
\begin{align}\label{eqn: suprema of chaos prob bound}
\|N_{\mset}(\bxi)\|_{L_p}\leq C(E+\sqrt{p}d_{\ttwo}(\mset)),
\end{align}
where $C$ is a constant depends only on $L$.
\end{theorem}

Here, $N_{\mset}(\bxi)$ represents the supremum of certain stochastic processes indexed by a set of matrices $\mA$. The above Proposition implies that $N_{\mset}(\bxi)$ can be bounded by three parameters: the suprema of Frobenius norms $d_F(\mset)$, the suprema of operator norms $d_{\ttwo}(\mset)$ and a $\gamma_2$-functional $\gamma_2(\mset,\|\cdot\|_{\ttwo})$, which can be bounded in terms of the covering numbers $N(\mset,\|\cdot\|_{\ttwo},u)$ as below.
\begin{align*}
\gamma_2(\mset,\|\cdot\|_{\ttwo})\leq c\int_{0}^{d_{\ttwo}(\mset)}\sqrt{\log N(\mset,\|\cdot\|_{\ttwo},u)} \ \mathrm{d}u,
\end{align*}
where the integral is known as Dudley integral or entropy integral \cite{talagrand2005generic}.

We can transfer the estimates on the moment \eqref{eqn: suprema of chaos prob bound} to a tail bound by the standard estimate due to Markov's inequality (see \cite[Proposition 7.15]{foucart2013mathematical}).
\begin{prop}\label{prop: bounds on chaos processes}
Following the definitions in Theorem \ref{thm: bounds on chaos processes}, for $t\geq 1$,
\begin{align}
\Prob(N_{\mset}(\bxi)\geq CE+Cd_{\ttwo}(\mset)t)\leq\exp(-t^2).
\end{align}
\end{prop}

It can be observed that $\Delta$ can be expressed in the form of $N_{\mset}(\bxi)$, where $\S$ and $\mA$ are replaced with $\v$ and $\mA_{\v}$, respectively. Now, we only need to estimate the parameters $d_F(\vset)$, $d_{\ttwo}(\vset)$ and $\gamma_2(\vset,\|\cdot\|_{\ttwo})$ before bounding $\Delta$ by using Theorem \ref{prop: bounds on chaos processes}. Since $\vset$ is a set of vectors, we have $d_F(\vset)=d_{\ttwo}(\vset)$ and $\gamma_2(\vset,\|\cdot\|_{\ttwo})=\gamma_2(\vset,\|\cdot\|_2)$.

For any vector $\x\in\sset$, we denote by $\x^s$ the length-$s$ vector that retains only the non-zero elements in $\x$. And correspondingly for any vector $\b\in\C^n$, we denote by $\b^s$ the length-$s$ vector that retains only the elements that have the same indexes as those of the non-zero elements in $\x$. We have, for any $\v\in\vset$,
\begin{align}
\|\v\|_2&=\|\z^{*}\U\diag(\tB\x)\|_2\leq\|\z^{*}\U\|_2\|\tB\x\|_{\infty}\nonumber\\
        &=\sqrt{\frac{\tn}{m}}\|\z\|_2\max_{j\in[\tn]}\{|\langle \tB_{(j,:)},\x\rangle|\}\nonumber\\
        &=\sqrt{\frac{\tn}{m}}\|\z\|_2\max_{j\in[\tn]}\{|\langle \tB_{(j,:)}^{s},\x^{s}\rangle|\}\nonumber \\
        &\leq\sqrt{\frac{\tn}{m}}\mu(\tB)\sqrt{s}\|\x\|_2\|\z\|_2\leq\frac{1}{2}\sqrt{\frac{\tn}{m}}\mu(\tB)\sqrt{s},\nonumber
\end{align}
where the last inequality is due to $\|\x\|_2^2+\|\z\|_2^2=1$. Therefore,
\begin{align}\label{eqn: operator norm width 1}
d_F(\vset)=d_{\ttwo}(\vset)\leq \frac{1}{2}\sqrt{\frac{\tn}{m}}\mu(\tB)\sqrt{s}.
\end{align}

Following the same steps, we can alternatively obtain, for any $\v\in\vset$,
\begin{align}
\|\v\|_2&=\|\z^{*}\U\diag(\tB\x)\|_2\leq\|\z^{*}\U\|_{\infty}\|\tB\x\|_{2}\nonumber\\
        &\leq\frac{1}{2}\mu(\U)\sqrt{k}.\nonumber
\end{align}
This provides another upper bound
\begin{align}\label{eqn: operator norm width 2}
d_F(\vset)=d_{\ttwo}(\vset)\leq \frac{1}{2}\mu(\U)\sqrt{k}.
\end{align}

We note that both \eqref{eqn: operator norm width 1} and \eqref{eqn: operator norm width 2} are valid bounds, and they are not comparable to each other since the relationship between $s$ and $k$ is unknown. It will be clear later that both bounds are useful for computing the entropy integrals. In particular, \eqref{eqn: operator norm width 1} and \eqref{eqn: operator norm width 2} are used for computing $I_1$ and $I_2$ respectively (as in \eqref{eqn: gamma2 bound}).

Next, we bound $\gamma_2$-functional $\gamma_2(\vset,\|\cdot\|_2)$ by estimating the covering numbers $N(\vset,\|\cdot\|_2,u)$. The derivation is divided into two steps.

\emph{Step 1. Decompose $N(\vset,\|\cdot\|_2,u)$.} Let $\mD_1=\{\x\in\R^n:\|\x\|_2^2\leq 1,\|\x\|_0\leq s\}$ and define the semi-norm $\|\cdot\|_{K_1}$ as
\begin{align}
\|\x\|_{K_1}=\|\U\diag(\tB\x)\|_{\ttwo}\quad\quad \forall\x\in\R^n.
\end{align}
For the metric space $(\mD_1,\|\cdot\|_{K_1})$, we take $\overbar{\mD_1}$ to be a $\frac{u}{2}$-net of $\mD_1$ with $|\overbar{\mD_1}|=N(\mD_1,\|\cdot\|_{K_1},\frac{u}{2})$. Let $\mD_2=\{\z\in\R^m:\|\z\|_2^2\leq 1,\|\z\|_{0}\leq k\}$ and define the semi-norm $\|\cdot\|_{K_2}$ as
\begin{align}
\|\z\|_{K_2}=\|\tB^*\diag(\U^*\z)\|_{\ttwo}\quad\quad \forall\z\in\R^m.
\end{align}
For the metric space $(\mD_2,\|\cdot\|_{K_2})$, we take $\overbar{\mD_2}$ to be a $\frac{u}{2}$-net of $\mD_2$ with $|\overbar{\mD_2}|=N(\mD_2,\|\cdot\|_{K_2},\frac{u}{2})$.

Now, let $\overbar{\vset}=\{(\bar{\z}^*\U\diag(\tB\bar{\x}))^*:\bar{\x}\in\overbar{\mD_1}, \bar{\z}\in\overbar{\mD_2}\}$ and remark that $|\overbar{\vset}|\leq|\overbar{\mD_1}||\overbar{\mD_2}|$. It remains to show that for all $\v\in\vset$, there exists $\bar{\v}\in\bar{\vset}$ with $\|\v-\bar{\v}\|_2\leq u$.

For any $\v=(\z^*\U\diag(\tB\x))^*\in\vset$, there exist $\bar{\v}=(\bar{\z}^*\U\diag(\tB\bar{\x}))^*\in\bar{\vset}$ with $\bar{\x}\in\overbar{\mD_1}$ and $\bar{\z}\in\overbar{\mD_2}$ obeying $\|\x-\bar{\x}\|_{K_1}\leq\frac{u}{2}$ and $\|\z-\bar{\z}\|_{K_2}\leq\frac{u}{2}$. This gives
\begin{align*}
\|\v-\bar{\v}\|_2&=\|\z^*\U\diag(\tB\x)-\bar{\z}^*\U\diag(\tB\bar{\x})\|_2\\
                 &=\|\z^*\U\diag(\tB\x)-\z^*\U\diag(\tB\bar{\x})\\
                 &\quad\quad +\z^*\U\diag(\tB\bar{\x})-\bar{\z}^*\U\diag(\tB\bar{\x})\|_2\\
                 &\leq\|\z^*\U\diag(\tB(\x-\bar{\x}))\|_2+\|(\z-\bar{\z})^*\U\diag(\tB\bar{\x})\|_2\\
                 &=\|\z^*\U\diag(\tB(\x-\bar{\x}))\|_2+\|\bar{\x}^*\tB^*\diag(\U^*(\z-\bar{\z}))\|_2\\
                 &\leq\|\z\|_2\|\U\diag(\tB(\x-\bar{\x}))\|_{\ttwo}\\
                 &\quad\quad +\|\bar{\x}\|_2\|\tB^*\diag(\U^*(\z-\bar{\z}))\|_{\ttwo}\\
                 &\overset{(a)}{\leq}\|\x-\bar{\x}\|_{K_1}+\|\z-\bar{\z}\|_{K_2}\leq u,
\end{align*}
where $(a)$ is due to the fact that $\|\z\|_2\leq 1$ and $\|\bar{\x}\|_2\leq 1$.

Hence,
\begin{align*}
N(\vset,\|\cdot\|_2,u)&\leq |\overbar{\vset}|\\
                      &\leq N(\mD_1,\|\cdot\|_{K_1},u/2)N(\mD_2,\|\cdot\|_{K_2},u/2).
\end{align*}

The $\gamma_2$-functional $\gamma_2(\vset,\|\cdot\|_2)$ can now be estimated by
\begin{align}
\gamma_2(\vset,\|\cdot\|_2)&\leq c\int_{0}^{d_{\ttwo}(\mset)}\sqrt{\log N(\vset,\|\cdot\|_2,u)} \mathrm{d}u\nonumber\\
                           &\lesssim \underbrace{\int_{0}^{d_{\ttwo}(\mset)}\sqrt{\log N(\mD_1,\|\cdot\|_{K_1},u/2)} \ \mathrm{d}u}_{I_1}\nonumber\\
                           &\quad+\underbrace{\int_{0}^{d_{\ttwo}(\mset)}\sqrt{\log N(\mD_2,\|\cdot\|_{K_2},u/2)} \ \mathrm{d}u}_{I_2}. \label{eqn: gamma2 bound}
\end{align}

\emph{Step 2. Estimate the covering numbers $N(\mD_1,\|\cdot\|_{K_1},u/2)$ and $N(\mD_2,\|\cdot\|_{K_2},u/2)$ and the entropy integrals.} We estimate each covering number in two different ways. For small value of $u$, we use a volumetric argument. For large value of $u$, we use the Maurey method (\cite[Lemma 4.2]{krahmer2014suprema}, or \cite[Problem 12.9]{foucart2013mathematical}). Then, the resultant covering number estimates can be used to compute the entropy integrals $I_1$ and $I_2$. Similar techniques on the covering number estimation and the entropy integral computation have been used in the CS literature, i.e., \cite{rauhut2010compressive,krahmer2014suprema,eftekhari2012restricted,zhang2014modulated}.

From \cite[Equation (28)]{zhang2014modulated} and \eqref{eqn: operator norm width 1}, we have
\begin{align}\label{eqn: I1 bound}
I_1\lesssim \sqrt{\frac{s\tn}{m}}\mu(\tB)(\log s)(\log \tn).
\end{align}

It remains to estimate $N(\mD_2,\|\cdot\|_{K_2},u/2)$ and compute $I_2$.
\emph{$1)$ small $u$}. We observe that $\mD_2$ is a subset of the union of $\binom m k$ unit Euclidean balls $\Bset_2^{k}$,
\begin{align}
\Bset_2^{k}:=\{\z\in\R^m:\|\z\|_2\leq 1,|\supp(\z)|\leq k\}.
\end{align}
For any $\z\in\mD_2$,
\begin{align}
\|\z\|_{K_2}&=\|\tB^*\diag(\U^*\z)\|_{\ttwo}\leq \|\U^*\z\|_{\infty}\leq\max_{i\in[n]}|\langle\U_{(:,i)},\z\rangle|\nonumber\\
            &\leq\mu(\U)\|\z\|_1\leq\mu(\U)\sqrt{k}\|\z\|_2\leq\sqrt{\frac{k}{m}}\|\z\|_2, \label{eqn: K2 bound 1 intermidiate}
\end{align}
where the last step is due to the assumption that $\mu(\U)\sim\frac{1}{\sqrt{m}}$.
Therefore,
\begin{align}
N(\mD_2,\|\cdot\|_{K_2},u/2)&\leq\binom m k N(\Bset_2^{k},\|\cdot\|_{K_2},u/2)\nonumber\\
                            &\leq\binom m k N(\Bset_2^{k},\sqrt{\frac{k}{m}}\|\cdot\|_2,u/2)\nonumber\\
                            &\leq(\frac{em}{k})^k(1+4\sqrt{\frac{k}{m}}\frac{1}{u})^{k}, \label{eqn: K2 bound 1}
\end{align}
where the last inequality is an application of \cite[Proposition 10.1]{rauhut2010compressive} and \cite[Lemma C.5]{foucart2013mathematical}.

\emph{$2)$ large $u$}. For any $\z\in\mD_2$, we have $\|\z\|_1\leq\sqrt{k}\|\z\|_2\leq\sqrt{k}$, which gives
\begin{align*}
\mD_2\subset\sqrt{k}\Bset_1^m:=\{\z\in\R^m: \|\z\|\leq\sqrt{k}\}.
\end{align*}
Then,
\begin{align*}
N(\mD_2,\|\cdot\|_{K_2},u/2)&\leq N(\sqrt{k}\Bset_1^m,\|\cdot\|_{K_2},u/2)\\
                            &=N(\Bset_1^m,\|\cdot\|_{K_2},u/(2\sqrt{k})).
\end{align*}

Based on the Maurey method, for $0<u<\frac{1}{2}\mu(\U)\sqrt{k}$, the covering number can be estimated by \cite[Lemma 8.3]{rauhut2010compressive}
\begin{align}
\sqrt{\log N(\mD_2,\|\cdot\|_{K_2},u/2)}&\lesssim \sqrt{k}\mu(\U)\sqrt{\log \tn\log m}u^{-1}\nonumber\\
&\leq \sqrt{\frac{k}{m}}\sqrt{\log \tn\log m}u^{-1}.\label{eqn: K2 bound 2}
\end{align}

We note that the estimation based on Maurey method depends on the range of the parameter $u$ (see \cite[Lemma 8.3]{rauhut2010compressive}), which is the reason why we employ different bounds (\eqref{eqn: operator norm width 1} and \eqref{eqn: operator norm width 2}) when computing the entropy integrals $I_1$ and $I_2$.

We now combine the results \eqref{eqn: K2 bound 1} and \eqref{eqn: K2 bound 2} to estimate the entropy integral $I_2$: we apply the first bound for $0<u\leq\frac{1}{10}\sqrt{\frac{1}{m}}$, and the second bound for $\frac{1}{10}\sqrt{\frac{1}{m}}<u\leq d_{\ttwo}(\vset)=\frac{1}{2}\sqrt{\frac{k}{m}}$. It reveals that
\begin{align}\label{eqn: I2 bound}
I_2\lesssim \sqrt{\frac{k}{m}}\log \tn \log k.
\end{align}
Combine \eqref{eqn: gamma2 bound}, \eqref{eqn: I1 bound} and \eqref{eqn: I2 bound}
\begin{align}
\gamma_2(\vset,\|\cdot\|_2)&\lesssim \sqrt{\frac{s\tn}{m}}\mu(\tB)(\log s)(\log \tn)\nonumber\\
                           &\quad\quad +\sqrt{\frac{k}{m}}\log \tn \log k.
\end{align}

Finally, we are ready to complete the proof by applying Proposition \ref{prop: bounds on chaos processes}. For the assumption on $m$ and $p$, $\delta\in(0,1)$,
\begin{align*}
m&\geq c_1\delta^{-2}s\tn\mu^2(\tB)\log^2 s\log^2 \tn\\
m&\geq c_2\delta^{-2}k\log^2 k\log^2 \tn,
\end{align*}
we have, by \eqref{eqn: operator norm width 1},
\begin{align*}
d_{F}(\vset)=d_{\ttwo}(\vset)\lesssim \frac{\delta}{\log s\log \tn},\quad\gamma_2(\vset,\|\cdot\|_2)\lesssim \delta.
\end{align*}
By substituting the above results into Proposition \ref{prop: bounds on chaos processes} (let $t=\log s\log \tn$), one obtains
\begin{align}
\Prob(\sup_{\v\in\vset} |\langle\v,\bxi\rangle|\leq c\delta)\geq 1-\exp(-\log^2 s\log^2 \tn).
\end{align}
The proof is completed by incorporating the constant $c$ into $c_1$, $c_2$.

\section{Proof of Theorem \ref{thm: skRIP of rdn unitary mtx}}\label{app: proof of skRIP of rdn unitary mtx}
Recall that in the measurement model $\y=\A\xs+\H\zs+\w$, $\A=\sqrt{\frac{n}{M}}\mR_{\Omega'}\G$ is a randomly sub-sampled unitary matrix and $\H\in\C^{M\times M}$ is a unitary matrix with $\mu(\H)\sim 1/\sqrt{M}$.

The following Lemma from \cite{rudelson2008sparse} is needed.
\begin{lemma}[Theorem 3.3 \cite{rudelson2008sparse}]\label{lem: rudelson08}
For the matrix $\A=\sqrt{\frac{n}{m}}\mR_{\Omega'}\G$, if for $\delta\in(0, 1)$,
\begin{align}
m\geq c\delta^{-2}s\log^4 n,
\end{align}
then with probability at least $1-n^{-\log^3 n}$ the restricted isometry constant $\delta_s$ of $\A$ satisfies $\delta_s\leq\delta$.
\end{lemma}

The $(s,k)$-RIP associated with $\TTheta=[\A,\  \H]$ can be  bounded by
\begin{align}\label{eqn: ripc decomposition}
\skripc&\leq \underbrace{\sup_{(\x,\z)\in\mT}\left|\|\A\x\|_2^2-\|\x\|_2^2\right|}_{\delta_1} +\underbrace{2\sup_{(\x,\z)\in\mT}|\langle\A\x,\H\z\rangle|}_{\delta_2},
\end{align}
where $\mT:=\{(\x,\z):\|\x\|_2^2+\|\z\|_2^2=1,\|\x\|_0\leq s,\|\z\|_{0}\leq k, \x\in\R^{n}, \z\in\R^{m}\}$.

By Lemma \ref{lem: rudelson08}, we have $\delta_1\leq\delta/2$ holds with probability $1-n^{-\log^3 n}$ for any $\delta\in(0,1)$ if $m\geq c\delta^{-2}s\log^4 n$.

Define a random vector $\d\in\R^n$ with i.i.d. entries satisfying $\lambda_i=\sqrt{\frac{m(n-m)}{n^2}}d_i+\frac{m}{n}$. Assume $\LLambda=\diag(\llambda)$ and $\H^*\mR_{\Omega'}=\U'$. We have,
\begin{align*}
\delta_2 &=2\sup_{(\x,\z)\in\mT}|\langle\sqrt{\frac{n}{M}}\mR_{\Omega'}\G\x,\H\z\rangle|\\
             &=2\sup_{(\x,\z)\in\mT}\left|\sqrt{\frac{n}{M}}\z^*\U'\LLambda\G\x\right|\\
             &=2\sup_{(\x,\z)\in\mT}\left|\sqrt{\frac{n}{M}}\z^*\U'\diag(\G\x)\llambda\right|\\
             &\leq 2\underbrace{\sup_{(\x,\z)\in\mT}\left|\frac{1}{2}\sqrt{\frac{n}{M}}\z^*\U'\diag(\G\x)\d\right|}_{t_1}\\
             &\quad +2\underbrace{\sup_{(\x,\z)\in\mT}\left|\sqrt{\frac{m^2}{Mn}}\z^*\U'\G\x\right|}_{t_2},
\end{align*}
where the last inequality is due to the fact that $\sqrt{\frac{m(n-m)}{n^2}}\leq\frac{1}{2}$ for any $m\leq n$.

Since $\llambda$ is a random Bernoulli vector with i.i.d. entries, by construction $\d$ is a length-$n$ random vector with independent, zero-mean, unit-variance, and $L$-subgaussian entries. Hence, the bound for $t_1$ can be formulated as the supremum of a stochastic process with the index $\mathcal{A}_{\r}$, where $\r=\sqrt{\frac{n}{M}}\z^*\U'\diag(\G\x)$ and $\mathcal{A}_{\r}:=\{\r:\|\x\|_2^2+\|\z\|_2^2=1,\|\x\|_0\leq s,\|\z\|_{0}\leq k\}$. For any $\r\in\mathcal{A}_{\r}$,
\begin{align*}
\|\r\|_2&=\sqrt{\frac{n}{M}}\|\z^{*}\U'\diag(\G\x)\|_2 \\
			&\leq \sqrt{\frac{n}{M}}\|\z^{*}\H^*\mR_{\Omega'}\|_2\|\G\x\|_{\infty}\\
        &=\sqrt{\frac{n}{M}}\|\z\|_2\max_{j\in[n]}\{|\langle \G_{(j,:)},\x\rangle|\}\nonumber\\
        &\leq\frac{1}{2}\sqrt{\frac{n}{M}}\mu(\G)\sqrt{s}, \nonumber\\
\|\r\|_2 &= \sqrt{\frac{n}{M}}\|\z^{*}\U'\diag(\G\x)\|_2\\
            &\leq \sqrt{\frac{n}{M}}\|\z^{*}\H^*\|_{\infty}\|\mR_{\Omega'}\G\x\|_{2} \\
            &\leq \sqrt{\frac{n}{M}}\mu(\H)\sqrt{k}\|\G_{(\Omega', :)}\x\|_{2} \\
            &\leq \frac{1}{2}\sqrt{n}\mu(\G)\mu(\H)\sqrt{k}.
\end{align*}
Therefore,
\begin{align*}
d_F(\mathcal{A}_{\r})&\leq \frac{1}{2}\sqrt{\frac{n}{M}}\mu(\G)\sqrt{s},\\
d_F(\mathcal{A}_{\r})&\leq \frac{1}{2}\sqrt{n}\mu(\G)\mu(\H)\sqrt{k}. \\
\end{align*}

By following the same proof steps as in Appendix \ref{appdix: proof of useful lemma}, we have
\begin{align}
\Prob(\sup_{\r\in\mathcal{A}_{\r}}|\langle\r,\d\rangle|\leq c\delta)\geq 1-\exp(-\log^2 s\log^2 n)
\end{align}
provided that
\begin{align*}
M&\geq c \delta^{-2}s n \mu^2(\G)\log^2 s\log^2 n \\
M&\geq c \delta^{-2}k n \mu^2(\G)\log^2 k\log^2 n.
\end{align*}

Bernstein's inequality \cite[Theorem A.1.13]{alon2004probabilistic} gives, for any $\nu>0$,
\begin{align}
\Prob(M>(1-\nu)m)\geq 1-\exp\left(-\frac{m\nu^2}{2}\right).
\end{align}

Hence, if
\begin{align*}
m &\geq \frac{1}{1-\nu} c \delta^{-2}s n \mu^2(\G)\log^2 s\log^2 n \\
m &\geq \frac{1}{1-\nu} c \delta^{-2}k n \mu^2(\G)\log^2 k\log^2 n,
\end{align*}
then
\begin{align*}
\Prob(\sup_{\r\in\mathcal{A}_{\r}}|\langle\r,\d\rangle|\leq c\delta)&\geq 1-\exp(-\log^2 s\log^2 n)\\
&\quad\quad -\exp\left(-\frac{m\nu^2}{2}\right).
\end{align*}

By assuming $m\geq 2c' \log n$ and $\nu=\sqrt{\frac{2c' \log n}{m}}$, the above probability of success can be written as
\begin{align}
\Prob(\sup_{\r\in\mathcal{A}_{\r}}|\langle\r,\d\rangle|\leq c\delta)\geq 1-n^{-\log^2 s\log n}-n^{-c'}.
\end{align}

For the second term, we have
\begin{align*}
t_2&=\sup_{(\x,\z)\in\mT}\left|\frac{m^2}{Mn}\z^*\U'\G\x\right|\\
     &\leq \frac{m^2}{Mn}\|\z\|_2\|\x\|_2\\
     &\leq \frac{1}{2}\frac{m^2}{Mn},
\end{align*}
where the last inequality is due to $\|\x\|_2^2+\|\z\|_2^2=1$. Therefore, $t_2\leq \delta/2$ for any $\delta\in(0,1)$ if $M\sim m$ and $m\leq \delta n$. By Bernstein's inequality, this condition can be satisfied with probability exceeding $1-n^{-c'}$ as long as $m\leq c\delta^2 n$. Theorem \ref{thm: skRIP of rdn unitary mtx} is proved by combining the above results.

\bibliographystyle{IEEEtran}
\bibliography{bibfile}

\begin{thebibliography}{10}
\providecommand{\url}[1]{#1}
\csname url@samestyle\endcsname
\providecommand{\newblock}{\relax}
\providecommand{\bibinfo}[2]{#2}
\providecommand{\BIBentrySTDinterwordspacing}{\spaceskip=0pt\relax}
\providecommand{\BIBentryALTinterwordstretchfactor}{4}
\providecommand{\BIBentryALTinterwordspacing}{\spaceskip=\fontdimen2\font plus
\BIBentryALTinterwordstretchfactor\fontdimen3\font minus
  \fontdimen4\font\relax}
\providecommand{\BIBforeignlanguage}[2]{{%
\expandafter\ifx\csname l@#1\endcsname\relax
\typeout{** WARNING: IEEEtran.bst: No hyphenation pattern has been}%
\typeout{** loaded for the language `#1'. Using the pattern for}%
\typeout{** the default language instead.}%
\else
\language=\csname l@#1\endcsname
\fi
#2}}
\providecommand{\BIBdecl}{\relax}
\BIBdecl

\bibitem{EJC05robustuncertainty}
E.~Candes, J.~Romberg, and T.~Tao, ``Robust uncertainty principles: {E}xact
  signal reconstruction from highly incomplete frequency information,''
  \emph{IEEE Trans. Inf. Theory}, vol.~52, no.~2, pp. 489--509, Feb 2006.

\bibitem{EJC06nearoptimal}
E.~Candes and T.~Tao, ``Near-optimal signal recovery from random projections:
  Universal encoding strategies?'' \emph{IEEE Trans. Inf. Theory}, vol.~52,
  no.~12, pp. 5406 --5425, Dec. 2006.

\bibitem{Donoho06compressedsensing}
D.~L. Donoho, ``Compressed sensing,'' \emph{IEEE Trans. Inf. Theory}, vol.~52,
  no.~4, pp. 1289--1306, Apr. 2006.

\bibitem{eldar2012compressed}
Y.~C. Eldar and G.~Kutyniok, \emph{Compressed sensing: theory and
  applications}.\hskip 1em plus 0.5em minus 0.4em\relax Cambridge University
  Press, 2012.

\bibitem{foucart2013mathematical}
S.~Foucart and H.~Rauhut, \emph{A mathematical introduction to compressive
  sensing}.\hskip 1em plus 0.5em minus 0.4em\relax Springer, 2013.

\bibitem{rauhut2010compressive}
H.~Rauhut, ``Compressive sensing and structured random matrices,''
  \emph{Theoretical foundations and numerical methods for sparse recovery},
  vol.~9, pp. 1--92, 2010.

\bibitem{zhang2014modulated}
P.~Zhang, L.~Gan, S.~Sun, and C.~Ling, ``Modulated unit-norm tight frames for
  compressed sensing,'' \emph{IEEE Trans. Signal Process.}, vol.~63, no.~15,
  pp. 3974--3985, Aug 2015.

\bibitem{bourgain2014improved}
J.~Bourgain, ``An improved estimate in the restricted isometry problem,'' in
  \emph{Geometric Aspects of Functional Analysis}.\hskip 1em plus 0.5em minus
  0.4em\relax Springer, 2014, pp. 65--70.

\bibitem{haviv2017restricted}
I.~Haviv and O.~Regev, ``The restricted isometry property of subsampled fourier
  matrices,'' in \emph{Geometric Aspects of Functional Analysis}.\hskip 1em
  plus 0.5em minus 0.4em\relax Springer, 2017, pp. 163--179.

\bibitem{li2013compressed}
X.~Li, ``Compressed sensing and matrix completion with constant proportion of
  corruptions,'' \emph{Constructive Approximation}, vol.~37, no.~1, pp. 73--99,
  2013.

\bibitem{Nguyen2013exact}
N.~Nguyen and T.~Tran, ``Exact recoverability from dense corrupted observations
  via $l_1$-minimization,'' \emph{IEEE Trans. Inf. Theory}, vol.~59, no.~4, pp.
  2017--2035, April 2013.

\bibitem{candes2011probabilistic}
E.~J. Candes and Y.~Plan, ``A probabilistic and {R}{I}{P}less theory of
  compressed sensing,'' \emph{IEEE Trans. Inf. Theory}, vol.~57, no.~11, pp.
  7235--7254, 2011.

\bibitem{adcock2017compressed}
B.~Adcock, A.~Bao, A.~Narayan \emph{et~al.}, ``Compressed sensing with sparse
  corruptions: Fault-tolerant sparse collocation approximations,'' \emph{arXiv
  preprint arXiv:1703.00135}, 2017.

\bibitem{krahmer2014suprema}
F.~Krahmer, S.~Mendelson, and H.~Rauhut, ``Suprema of chaos processes and the
  restricted isometry property,'' \emph{Communications on Pure and Applied
  Mathematics}, 2014.

\bibitem{haupt2008compressed}
J.~Haupt, W.~U. Bajwa, M.~Rabbat, and R.~Nowak, ``Compressed sensing for
  networked data,'' \emph{IEEE Signal Process. Mag.}, vol.~25, no.~2, pp.
  92--101, 2008.

\bibitem{franke2009orden}
C.~Franke and M.~Gertz, ``{O}{R}{D}{E}{N}: {O}utlier region detection and
  exploration in sensor networks,'' in \emph{Proceedings of the 2009 ACM SIGMOD
  International Conference on Management of data}.\hskip 1em plus 0.5em minus
  0.4em\relax ACM, 2009, pp. 1075--1078.

\bibitem{zhang2010outlier}
Y.~Zhang, N.~Meratnia, and P.~Havinga, ``Outlier detection techniques for
  wireless sensor networks: A survey,'' \emph{Commun. Surveys Tuts.}, vol.~12,
  no.~2, pp. 159--170, 2010.

\bibitem{charbi2010compressive}
Z.~Charbiwala, S.~Chakraborty, S.~Zahedi, Y.~Kim, M.~Srivastava, T.~He, and
  C.~Bisdikian, ``Compressive oversampling for robust data transmission in
  sensor networks,'' in \emph{INFOCOM, 2010 Proceedings IEEE}, March 2010, pp.
  1--9.

\bibitem{li2010systematic}
Z.~Li, F.~Wu, and J.~Wright, ``On the systematic measurement matrix for
  compressed sensing in the presence of gross errors,'' in \emph{Data
  Compression Conference (DCC), 2010}.\hskip 1em plus 0.5em minus 0.4em\relax
  IEEE, 2010, pp. 356--365.

\bibitem{laska2009exact}
J.~N. Laska, M.~Davenport, R.~G. Baraniuk \emph{et~al.}, ``Exact signal
  recovery from sparsely corrupted measurements through the pursuit of
  justice,'' in \emph{Signals, Systems and Computers, 2009 Conference Record of
  the Forty-Third Asilomar Conference on}.\hskip 1em plus 0.5em minus
  0.4em\relax IEEE, 2009, pp. 1556--1560.

\bibitem{studer2012recovery}
C.~Studer, P.~Kuppinger, G.~Pope, and H.~Bolcskei, ``Recovery of sparsely
  corrupted signals,'' \emph{IEEE Trans. Inf. Theory}, vol.~58, no.~5, pp.
  3115--3130, 2012.

\bibitem{popovic1991synthesis}
B.~M. Popovic, ``Synthesis of power efficient multitone signals with flat
  amplitude spectrum,'' \emph{IEEE Trans. Commun.}, vol.~39, no.~7, pp.
  1031--1033, 1991.

\bibitem{davis1999peak}
J.~A. Davis and J.~Jedwab, ``Peak-to-mean power control in {O}{F}{D}{M},
  {G}olay complementary sequences, and {R}eed-{M}uller codes,'' \emph{IEEE
  Trans. Inf. Theory}, vol.~45, no.~7, pp. 2397--2417, 1999.

\bibitem{levitt1985fh}
B.~Levitt, ``{F}{H}/{M}{F}{S}{K} performance in multitone jamming,'' \emph{IEEE
  J. Sel. Areas Commun.}, vol.~3, no.~5, pp. 627--643, 1985.

\bibitem{berger2010sparse}
C.~R. Berger, S.~Zhou, J.~C. Preisig, and P.~Willett, ``Sparse channel
  estimation for multicarrier underwater acoustic communication: From subspace
  methods to compressed sensing,'' \emph{IEEE Trans. Signal Process.}, vol.~58,
  no.~3, pp. 1708--1721, 2010.

\bibitem{haupt2010toeplitz}
J.~Haupt, W.~U. Bajwa, G.~Raz, and R.~Nowak, ``Toeplitz compressed sensing
  matrices with applications to sparse channel estimation,'' \emph{IEEE Trans.
  Inf. Theory}, vol.~56, no.~11, pp. 5862--5875, 2010.

\bibitem{meng2012compressive}
J.~Meng, W.~Yin, Y.~Li, N.~T. Nguyen, and Z.~Han, ``Compressive sensing based
  high-resolution channel estimation for {O}{F}{D}{M} system,'' \emph{IEEE J.
  Sel. Topics Signal Process.}, vol.~6, no.~1, pp. 15--25, 2012.

\bibitem{li2013convolutional}
K.~Li, L.~Gan, and C.~Ling, ``Convolutional compressed sensing using
  deterministic sequences,'' \emph{IEEE Trans. Signal Process.}, vol.~61,
  no.~3, pp. 740--752, 2013.

\bibitem{umehara2006performance}
D.~Umehara, H.~Nishiyori, and Y.~Morihiro, ``Performance evaluation of
  {C}{M}{F}{B} transmultiplexer for broadband power line communications under
  narrowband interference,'' in \emph{Power Line Communications and Its
  Applications, 2006 IEEE International Symposium on}.\hskip 1em plus 0.5em
  minus 0.4em\relax IEEE, 2006, pp. 50--55.

\bibitem{gomaa2011sparsity}
A.~Gomaa and N.~Al-Dhahir, ``A sparsity-aware approach for {N}{B}{I} estimation
  in {M}{I}{M}{O}-{O}{F}{D}{M},'' \emph{IEEE Trans. Wireless Commun.}, vol.~10,
  no.~6, pp. 1854--1862, 2011.

\bibitem{romberg2009compressive}
J.~Romberg, ``Compressive sensing by random convolution,'' \emph{SIAM J.
  Imaging Sciences}, vol.~2, no.~4, pp. 1098--1128, 2009.

\bibitem{puy2012universal}
G.~Puy, P.~Vandergheynst, R.~Gribonval, and Y.~Wiaux, ``Universal and efficient
  compressed sensing by spread spectrum and application to realistic {F}ourier
  imaging techniques,'' \emph{EURASIP J. Advances in Signal Processing}, vol.
  2012, no.~1, pp. 1--13, 2012.

\bibitem{tropp2010beyond}
J.~A. Tropp, J.~N. Laska, M.~F. Duarte, J.~K. Romberg, and R.~G. Baraniuk,
  ``Beyond {N}yquist: Efficient sampling of sparse bandlimited signals,''
  \emph{IEEE Trans. Inf. Theory}, vol.~56, no.~1, pp. 520--544, 2010.

\bibitem{romberg2010sparse}
J.~Romberg and R.~Neelamani, ``Sparse channel separation using random probes,''
  \emph{Inverse Problems}, vol.~26, no.~11, p. 115015, 2010.

\bibitem{slavinsky2011compressive}
J.~P. Slavinsky, J.~N. Laska, M.~A. Davenport, and R.~G. Baraniuk, ``The
  compressive multiplexer for multi-channel compressive sensing,'' in
  \emph{IEEE Int. Conf. Acoustics, Speech and Signal Processing (ICASSP)},
  2011, pp. 3980--3983.

\bibitem{rivenson2010single}
Y.~Rivenson, A.~Stern, and B.~Javidi, ``Single exposure super-resolution
  compressive imaging by double phase encoding,'' \emph{Optics express},
  vol.~18, no.~14, pp. 15\,094--15\,103, 2010.

\bibitem{li2015compressiveoptical}
J.~Li, J.~S. Li, Y.~Y. Pan, and R.~Li, ``Compressive optical image
  encryption,'' \emph{Scientific reports}, vol.~5, 2015.

\bibitem{herman2009high}
M.~A. Herman and T.~Strohmer, ``High-resolution radar via compressed sensing,''
  \emph{IEEE Trans. Signal Process.}, vol.~57, no.~6, pp. 2275--2284, 2009.

\bibitem{marcia2008fast}
R.~F. Marcia, C.~Kim, J.~Kim, D.~J. Brady, and R.~M. Willett, ``Fast
  disambiguation of superimposed images for increased field of view,'' in
  \emph{Image Processing, 2008. ICIP 2008. 15th IEEE International Conference
  on}.\hskip 1em plus 0.5em minus 0.4em\relax IEEE, 2008, pp. 2620--2623.

\bibitem{pham2012improved}
D.-S. Pham and S.~Venkatesh, ``Improved image recovery from compressed data
  contaminated with impulsive noise,'' \emph{IEEE Trans. Image Process.},
  vol.~21, no.~1, pp. 397--405, 2012.

\bibitem{filipovic2014reconstruction}
M.~Filipovi{\'c}, ``Reconstruction of sparse signals from highly corrupted
  measurements by nonconvex minimization,'' in \emph{IEEE International
  Conference on Acoustics, Speech and Signal Processing (ICASSP)}, 2014.

\bibitem{saab2008stable}
R.~Saab, R.~Chartrand, and O.~Yilmaz, ``Stable sparse approximations via
  nonconvex optimization,'' in \emph{Acoustics, Speech and Signal Processing,
  2008. ICASSP 2008. IEEE International Conference on}.\hskip 1em plus 0.5em
  minus 0.4em\relax IEEE, 2008, pp. 3885--3888.

\bibitem{chartrand2008iteratively}
R.~Chartrand and W.~Yin, ``Iteratively reweighted algorithms for compressive
  sensing,'' in \emph{IEEE International Conference on Acoustics, Speech and
  Signal Processing (ICASSP)}.\hskip 1em plus 0.5em minus 0.4em\relax IEEE,
  2008, pp. 3869--3872.

\bibitem{Foygel2014corrupted}
R.~Foygel and L.~Mackey, ``Corrupted sensing: {N}ovel guarantees for separating
  structured signals,'' \emph{IEEE Trans. Inf. Theory}, vol.~60, no.~2, pp.
  1223--1247, Feb 2014.

\bibitem{chen2013robust}
Y.~Chen, C.~Caramanis, and S.~Mannor, ``Robust sparse regression under
  adversarial corruption,'' in \emph{Proceedings of the 30th International
  Conference on Machine Learning (ICML-13)}, 2013, pp. 774--782.

\bibitem{studer2014stable}
C.~Studer and R.~G. Baraniuk, ``Stable restoration and separation of
  approximately sparse signals,'' \emph{Applied and Computational Harmonic
  Analysis}, vol.~37, no.~1, pp. 12--35, 2014.

\bibitem{pope2013probabilistic}
G.~Pope, A.~Bracher, and C.~Studer, ``Probabilistic recovery guarantees for
  sparsely corrupted signals,'' \emph{IEEE Trans. Inf. Theory}, vol.~59, no.~5,
  pp. 3104--3116, May 2013.

\bibitem{wright2009robust}
J.~Wright, A.~Y. Yang, A.~Ganesh, S.~S. Sastry, and Y.~Ma, ``Robust face
  recognition via sparse representation,'' \emph{IEEE Trans. Pattern Anal.
  Mach. Intell.}, vol.~31, no.~2, pp. 210--227, 2009.

\bibitem{laska2011democracy}
J.~N. Laska, P.~T. Boufounos, M.~A. Davenport, and R.~G. Baraniuk, ``Democracy
  in action: Quantization, saturation, and compressive sensing,'' \emph{Applied
  and Computational Harmonic Analysis}, vol.~31, no.~3, pp. 429--443, 2011.

\bibitem{chiconvex}
Y.~Chi, ``Convex relaxations of spectral sparsity for robust super-resolution
  and line spectrum estimation,'' in \emph{Wavelets and Sparsity XVII}, vol.
  10394.\hskip 1em plus 0.5em minus 0.4em\relax International Society for
  Optics and Photonics, 2017, p. 103941G.

\bibitem{fernandez2016demixing}
C.~Fernandez-Granda, G.~Tang, X.~Wang, and L.~Zheng, ``Demixing sines and
  spikes: Robust spectral super-resolution in the presence of outliers,''
  \emph{arXiv preprint arXiv:1609.02247}, 2016.

\bibitem{cvx}
M.~Grant and S.~Boyd, ``{CVX}: Matlab software for disciplined convex
  programming, version 2.1,'' \url{http://cvxr.com/cvx}, Mar. 2014.

\bibitem{gb08}
------, ``Graph implementations for nonsmooth convex programs,'' in
  \emph{Recent Advances in Learning and Control}, ser. Lecture Notes in Control
  and Information Sciences, V.~Blondel, S.~Boyd, and H.~Kimura, Eds.\hskip 1em
  plus 0.5em minus 0.4em\relax Springer-Verlag Limited, 2008, pp. 95--110,
  \url{http://stanford.edu/~boyd/graph_dcp.html}.

\bibitem{talagrand2005generic}
M.~Talagrand, \emph{The generic chaining}.\hskip 1em plus 0.5em minus
  0.4em\relax Springer, 2005, vol. 154.

\bibitem{eftekhari2012restricted}
A.~Eftekhari, H.~L. Yap, C.~J. Rozell, and M.~B. Wakin, ``The restricted
  isometry property for random block diagonal matrices,'' \emph{arXiv preprint
  arXiv:1210.3395}, 2012.

\bibitem{rudelson2008sparse}
M.~Rudelson and R.~Vershynin, ``On sparse reconstruction from {F}ourier and
  {G}aussian measurements,'' \emph{Communications on Pure and Applied
  Mathematics}, vol.~61, no.~8, pp. 1025--1045, 2008.

\bibitem{alon2004probabilistic}
N.~Alon and J.~H. Spencer, \emph{The probabilistic method}.\hskip 1em plus
  0.5em minus 0.4em\relax John Wiley \& Sons, 2004.

\end{thebibliography}
\end{document}